\documentclass[12pt]{article} 
\usepackage[colorlinks, citecolor={blue}]{hyperref}
\usepackage{url}
\usepackage{amsfonts,amscd,amssymb}
\usepackage{amsthm,amsmath,natbib}
\usepackage{algorithmic,algorithm}
\usepackage{bm}
\usepackage{bbm} 
\usepackage{color}
\usepackage{verbatim}
\usepackage{graphicx}
\usepackage{setspace}
\usepackage{natbib}
\usepackage[margin=1in]{geometry}
\usepackage{enumitem}
\usepackage[textsize=tiny]{todonotes}

\usepackage{tikz}
\usetikzlibrary{matrix}
\usetikzlibrary{backgrounds}
\usetikzlibrary{calc}
\usetikzlibrary{arrows,shapes}
\usetikzlibrary{decorations.pathmorphing}

\newtheorem{thm}{Theorem}
\newtheorem{lemma}{Lemma}

\newtheorem{cor}{Corollary}

\newtheorem{mydef}{Definition}
\newtheorem{assumption}{Assumption}

\newcommand{\mb}{\mathbf}
\newcommand{\mbb}{\mathbb}
\newcommand{\mc}{\mathcal}

\definecolor{trevorblue}{rgb}{0.330, 0.484, 0.828}
\definecolor{trevoryellow}{rgb}{0.829, 0.680, 0.306}

\newcommand{\comp}{\mc{C}}
\newcommand{\rate}{\mu}
\newcommand{\ratemat}{\mathbf{R}}  
\newcommand{\map}{\mc{A}}
\newcommand{\pars}{\mathbf \theta}
\newcommand{\upper}{\mc{U}}

\title{Direct likelihood-based inference for discretely observed stochastic compartmental models of infectious disease}
\date{}
\author{
Lam Si Tung Ho \\
Department of Biostatistics \\
University of California, Los Angeles
\and
Forrest W.~Crawford \\
Departments of Biostatistics and Ecology \& Evolutionary Biology \\
Yale University
\and
Marc A.~Suchard \\
Departments of Biomathematics, Biostatistics and Human Genetics \\
University of California, Los Angeles
}

\begin{document}
\maketitle

\clearpage

\begin{abstract}
Stochastic compartmental models are important tools for understanding the course of infectious diseases epidemics in populations and in prospective evaluation of intervention policies.
However, calculating the likelihood for discretely observed data from even simple models -- such as the ubiquitous susceptible-infectious-removed (SIR) model -- has been considered computationally intractable, since its formulation almost a century ago.
Recently researchers have proposed methods to circumvent this limitation through data augmentation or approximation, but these approaches often suffer from high computational cost or loss of accuracy.
We develop the mathematical foundation and an efficient algorithm to compute the likelihood for discretely observed data from a broad class of stochastic compartmental models.
We also give expressions for the derivatives of the transition probabilities using the same technique, making possible inference via Hamiltonian Monte Carlo (HMC).
We use the 17th century plague in Eyam, a classic example of the SIR model, to compare our recursion method to sequential Monte Carlo, analyze using HMC, and assess the model assumptions.
We also apply our direct likelihood evaluation to perform Bayesian inference for the 2014-2015 Ebola outbreak in Guinea.
The results suggest that the epidemic infectious rates have decreased since October 2014 in the Southeast region of Guinea, while rates remain  the same in other regions, facilitating understanding of the outbreak and the effectiveness of Ebola control interventions.

\end{abstract}

\clearpage

\section{Introduction}
Compartmental models have been used extensively in epidemiology to study the spread of infectious diseases such as plague \citep{raggett1982}, measles \citep{cauchemez2008likelihood}, influenza \citep{dukic2012tracking}, HIV \citep{blum2010hiv}, and Ebola \citep{althaus2014estimating}.
These models stratify the population into separate groups according to differing health states.
The famous susceptible-infectious-removed (SIR) model \citep{mckendrick1926applications, Kermack1927} divides the population into three subpopulations: the susceptible (S) group including healthy persons who have no immunity to the disease, the infectious (I) group including infected persons who can transmit the disease to susceptible persons by contact, and the removed (R) group including recovered/dead persons who no longer affect disease dynamics.
Important adaptions of the SIR model abound.
For example, allowing for the loss of immunity in the removed group such that recovered persons can become susceptible again results in the susceptible-infectious-removed-susceptible (SIRS) model.  As a simplification, the susceptible-infectious-susceptible (SIS) model assumes that individuals who recover from the disease have no immunity against reinfection, thus rejoin susceptible group immediately after recovery.
The more complicated susceptible-exposed-infectious-removed (SEIR) model takes into account an incubation period by adding an exposed (E) group including individuals who are infected but not yet infectious.

Compartmental models have been studied in both deterministic and stochastic settings.
One advantage of deterministic models is that they yield simpler statistical inference than their stochastic counterparts.
 However, ``many infectious disease systems are fundamentally individual-based stochastic processes, and are more naturally described by stochastic models'' \citep{roberts2015nine}.
Deterministic models are only appropriate when the populations of the compartments are sufficiently large \citep{brauer2008compartmental}.
Therefore, stochastic models remain preferable when their analysis is possible. If we are able to observe all transition events, likelihood-based inference for stochastic compartment models is straightforward.
For example, \citet{becker1999statistical} derive maximum likelihood estimates under complete observation for the SIR model. Unfortunately, it is very unlikely that we know exactly when an individual contracts the disease.
In general, surveillance data often include total counts of individuals in each compartment at several observation points.
Calculation of the likelihood requires evaluating the transition probabilities of the underlying stochastic process between these time points and, thus, becomes intractable due to the requirement of integrating over all unobserved events \citep{cauchemez2008likelihood}.
Solving for the transition probabilities begins, as \cite{renshaw2011} reminds us, by innocuously writing out the Chapman-Kolmogorov equations for the compartmental model, but the ``associated mathematical manipulations required to generate solutions can only be described as heroic.''

One common solution considers stochastic compartmental models as finite, but very large, state-space Markov processes and approximates their transition probabilities using matrix exponentiation.
Unfortunately, this method is extremely time consuming and numerically unstable in many instances \citep{schranz2008pathological, crawford2012}.
Further, when the state-space is infinite, matrix exponentiation can suffer from truncation error \citep{crawford2016coupling}.
Several alternative approaches have been developed to overcome the intractability of compartmental models, including data augmentation, diffusion approximation, sequential Monte Carlo (SMC)  -- namely, particle filters -- and approximate Bayesian computation (ABC).
However, these methods are limited and do not completely achieve tractability.
In Section \ref{sec:SCM}, we give a formal definition of stochastic compartmental models and discuss limitations of existing methods in more detail.

In this paper, we propose a
method with polynomial complexity
to compute the transition probabilities and their derivatives for stochastic compartmental models, making direct inference scalable to large epidemics. The main
technique
of our method is solving the Chapman-Kolmogorov equations in the Laplace domain and evaluating the inverse Laplace transform of these solutions numerically to get back the transition probabilities.
Recently, this
technique
has been successfully applied to the SIS model \citep{crawford2012} and the SIR model \citep{ho2016}, where the solutions of the Chapman-Kolmogorov equations in the Laplace domain can be represented by continued fractions.
Although these results make progress toward evaluating the likelihood function efficiently, applying the continued fraction representation for more complex models such as SEIR and SIRS remains an open problem.
In this work, we bypass the need for an exotic continued fraction representation by constructing multivariate birth processes that are equivalent to epidemic processes of the compartmental models.
Consequently, our method does not require evaluating continued fractions, and is therefore significantly faster and straightforward to apply to complex compartmental models.
Section \ref{sec:transprob} explains the construction of multivariate birth process representations and the dynamic programming algorithm for computing the transition probabilities of compartmental models.
In Section \ref{sec:epimodels}, we apply this new method to three prevailing infectious disease models (SIR, SEIR, and SIRS) and illustrate the computation gain for the SIR model compared to the method in \citet{ho2016}, the SMC method implemented in the increasingly popular \texttt{R} package \texttt{pomp} \citep{king2015}, and the matrix exponentiation method implemented in the state-of-the-art software \texttt{Expokit} \citep{sidje1998}.
We discuss two
%
further statistical applications using our recursion
%
%
which do not appear
possible under previous approaches
in Section \ref{sec:addapp}.
Specifically, we devise polynomial-time computable derivatives of the transition probabilities of the SIR model, enabling an analysis of the dynamics of an historical plague outbreak using Hamiltonian Monte Carlo (HMC).
Further, the generality of our method equips us to explore the adequacy of the SIR model assumptions for this outbreak of plague.
Finally, in Section \ref{sec:Ebola}, we turn to the 2014-2015 Ebola outbreak in Guinea and propose a time-inhomogeneous, hierarchical SIR extension that provides evidence for the slowing of this outbreak.  Moreover, we find that the change in the trajectory only happened in the Southeast region of Guinea.


\section{Stochastic compartmental models}
\label{sec:SCM}

In this section, we formally define stochastic compartmental models, discuss limitations of current inference methods when the data are observed discretely, and propose a new method of polynomial complexity for computing their transition probabilities.


\subsection{Notation and definition}

A stochastic $m$-compartmental model stratifies the population into $m$ homogeneous subpopulations called compartments.
Let
$\{ \comp_1, \comp_2, \ldots, \comp_m \}$
be the compartments and $\mb{Y}(t) = \{ Y_1(t),Y_2(t),\ldots,Y_m(t) \}$
be their population at time $t \ge 0$, then
the rate matrix  $\ratemat$ is an $m \times m$ matrix $[\rate_{ij}(\theta, \mb{Y})]_{1 \leq i,j \leq m}$ where $\rate_{ij}(\theta, \mb{Y}) \geq 0$ is a function of the parameter of interest $\pars$ and $\mb{Y}(t)$, representing an infinitesimal transition rate from $\comp_i$ to $\comp_j$.
We set $\rate_{ii}(\theta, \mb{Y}) = 0$ for all $i = 1,\ldots,m$.
Let $d$ count the number of positive elements of $\ratemat$.
Then, there are $d$ possible transitions of $\mb{Y}$ during a sufficient small time interval $(t, t+dt)$:
\begin{equation}
\begin{aligned}
 \Pr \left \{ \mb{Y}(t + dt) = \mb{y} - \mb{e}_i + \mb{e}_j ~|~ \mb{Y}(t) =\mb{y} \right \} & = \rate_{ij}(\theta, \mb{y}) dt + o(dt),~~\mu_{ij} \ne 0  \\
 \Pr \left \{ \mb{Y}(t + dt) = \mb{y}  ~|~ \mb{Y}(t) =\mb{y} \right \} & = 1 - \left ( \sum_{i,j=1}^m{ \rate_{ij}(\theta, \mb{y})} \right )dt + o(dt),
\end{aligned}
\label{eqn:bb}
\end{equation}
where $\mb{e}_i$ and $\mb{e}_j$ are the $i^{\mbox{\tiny th}}$ and $j^{\mbox{\tiny th}}$ coordinate vector of $\mathbb{R}^m$ respectively.
We call $\mb{Y}(t)$ a compartmental process.
We can visualize a compartmental model by a directed graph where nodes correspond to compartments and a directed edge from node $i$ to node $j$ means $\mu_{ij}$ is positive.
Figure \ref{fig:excomp} gives an example of representing a 3-compartmental model by a directed graph.

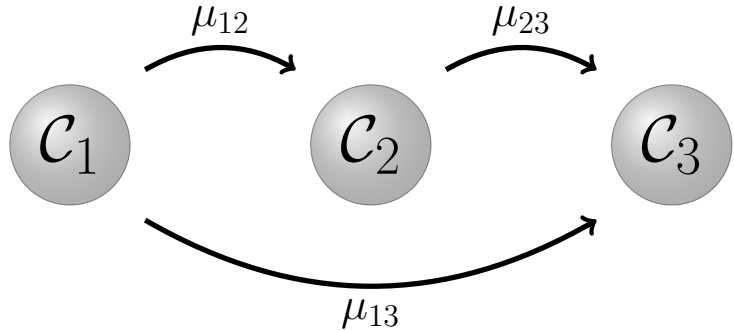
\begin{figure}[h]
\centering
\begin{tikzpicture}
%
%
\shadedraw[ball color=black!50, opacity=0.4] (0, 0) circle (0.8) node[opacity=1.0] {\Huge $\mc{C}_1$};
\shadedraw[ball color=black!50, opacity=0.4] (4, 0) circle (0.8) node[opacity=1.0] {\Huge $\mc{C}_2$};
\shadedraw[ball color=black!50, opacity=0.4] (8, 0) circle (0.8) node[opacity=1.0] {\Huge $\mc{C}_3$};
%
%
\draw[bend left,->, line width=2pt] (1, 1) to node[above] {{\Large $\mu_{12} $}} (3, 1);
\draw[bend left, ->, line width=2pt] (5, 1) to node[above] {{\Large $\mu_{23}$}} (7, 1);
\draw[bend right, ->, line width=2pt] (1, -1) to node[below] {{\Large $\mu_{13}$}} (7, -1);
\end{tikzpicture}
\caption{A directed graph representation of a 3-compartmental model. The rate matrix $\ratemat$ of this model only has $d = 3$ positive elements: $\mu_{12}$, $\mu_{23}$, and $\mu_{13}$.}
\label{fig:excomp}
\end{figure}


\subsection{Limitations of current approaches}


The first approach for likelihood-based inference under discretely-observed stochastic compartmental models exploits data augmentation.
This technique augments the observed data with the extensive unobserved information needed to evaluate the continuously-observed likelihood.
This method often treats the times of all unobserved events as parameters and explores the joint posterior distribution by Markov chain Monte Carlo (MCMC) method \citep{gibson1998estimating, o1999bayesian, o2002tutorial}.
Although data augmentation works well for small epidemics, it has been criticized for being computationally prohibitive with large augmented data \citep{cauchemez2008likelihood, blum2010hiv}.

An alternative approach to data augmentation entertains a diffusion approximation.
This method approximates the discrete compartmental processes by continuous diffusion processes whose likelihood function is easy to calculate.
For example, \citet{cauchemez2008likelihood} propose to mimic the SIR process by a Cox-Ingersoll-Ross process \citep{cox1985theory}, and apply this approximation to study measles epidemics in London (1948-1964).
However, a diffusion approximation is not applicable to epidemics in small communities because the approximation requires the state-space to be large enough to justify approximating a discrete process by a continuous one \citep{karev2005modeling, golightly2005bayesian}.
Moreover, this method is often not sufficiently accurate for use even as a simulator \citep{golightly2005bayesian}.

Particle filters, as a SMC approach, offer another popular tool for estimating the likelihood of stochastic models  \citep{arulampalam2002tutorial}.
The \texttt{R} package \texttt{pomp} \citep{king2015} provides an increasingly popular SMC implementation for both frequentist and Bayesian inference settings.
For example, \citet{ionides2006inference} develop an iterated filtering method that uses a particle filter to approximate the maximum likelihood estimates of the parameters.
In the Bayesian setting, \citet{andrieu2010particle} construct a particle marginal Metropolis-Hastings sampler to explore the posterior distribution using estimates from a particle filter.
The computational cost of these methods can be prohibitive when the convergence is slow because each iteration requires using a particle filter to estimate the likelihood \citep{owen2015scalable}.

Another alternative to data augmentation is ABC \citep{blum2010hiv}.
This is a likelihood-free approach replacing the observations with summary statistics and approximating the posterior of the parameters given the summary statistics by a simulation-based method.
Nonetheless, the ABC method can be biased because of non-zero tolerance and non-sufficient summary statistics \citep{sunnaaker2013approximate}, especially in high dimensions \citep{blum2010hiv}.
Therefore, credible interval estimates tend to be inflated \citep{csillery2010approximate}, and model selection using the ABC method cannot be trusted \citep{robert2011lack}.


Finally, \citet{faddy1977stochastic} proposes an approximation for the stochastic SIR model by assuming that each susceptible person becomes infected independently with the same rate $\beta \times i(t)$ where $i(t)$ is the number of infected individuals in the deterministic SIR model \citep{Kermack1927}.
The transition probabilities of this approximated process have analytic formulae because of the independence assumption, but this approximation becomes less accurate as the epidemic progresses.


\section{Evaluating transition probabilities}
\label{sec:transprob}

We present a new method for computing the transition probabilities of stochastic compartmental models.
Our method achieves polynomial complexity, thus enabling direct likelihood-based inference for discretely observed data.
The main idea is to recast a compartmental process whose rate matrix $\ratemat$ has $d$ positive elements into a $d$-dimensional birth process by keeping track of $d$ types of transition events between compartments.
This idea has been used in chemical thermodynamics for almost 100 years, where the variable measuring the progress of all substances in a chemical reaction is called the \emph{degree of advancement} or \emph{extent of reaction} variable \citep{dedonder1920lecons}.
By doing this, we can evaluate the transition probabilities more efficiently because the resulting multivariate birth processes are monotonically non-decreasing, while the compartment populations may increase or decrease over time.
This monotonicity affords us the opportunity to apply dynamic programming for building the transition probability matrix.

%

\subsection{Multivariate birth process}

\begin{mydef}
A $d$-dimensional birth process is a continuous-time Markov process counting the number of ``birth'' events for $d$ populations. Let $\mb{X}(t) = \{ X_1(t), X_2(t), \ldots, X_d(t) \},~t \geq 0$ be a multivariate birth process, whose state-space is $\mathbb{N}^d$. Then, there are $d + 1$ possible transitions of $\mb{X}$ during a sufficiently small time interval $(t,t+dt)$:
\begin{equation}
\begin{aligned}
 \Pr \left \{ \mb{X}(t + dt) = \mb{x} + \mb{e}_k ~|~ \mb{X}(t) =\mb{x} \right \} & = \lambda^{(k)}_{\mb{x}} dt + o(dt),~k \in \{ 1, 2, \ldots, d\}  \\
 \Pr \left \{ \mb{X}(t + dt) = \mb{x}  ~|~ \mb{X}(t) =\mb{x} \right \} & = 1 - \left ( \sum_{k=1}^d{\lambda^{(k)}_{\mb{x}}} \right )dt + o(dt),
\end{aligned}
\label{eqn:mbp}
\end{equation}
where
$\lambda^{(k)}_{\mb{x}} \geq 0$ is the birth rate of the $k^{\mbox{\tiny th}}$ population given the current population is $\mb{x} = (x_1, x_2, \ldots, x_d)$.
\end{mydef}

For two vectors $\mb{u}, \mb{v} \in \mathbb{N}^d$, denote $P_{\mb{u} \mb{v}}(t) = \Pr \{\mb{X}(t)= \mb{v}~ |~\mb{X}(0)=\mb{u}\}$ be the transition probability of the multivariate birth process from $\mb{u}$ to $\mb{v}$ after $t$ units of time.
We say $\mb{u} \leq \mb{v}$ if $u_k \leq v_k$ for every $k = 1, 2, \ldots, d$. Notice that $P_{\mb{u} \mb{v}}(t) \ne 0$ if and only if $\mb{u} \leq \mb{v}$.

Let $\mb{B} \in \mathbb{N}^d$, and set $\lambda^{(k)}_{\mb{x}} = 0$ if $x_k = -1$. For $i \in \mathbb{N}$, we denote
\begin{equation}
D_i = \left \{ \mb{x}: \sum_{k=1}^d{x_k} = i \right \},~~\text{and}~~ \lambda_i = \max_{\mb{x} \in D_i} \left \{ \sum_{k=1}^d{\lambda^{(k)}_{\mb{x}}} \right \}.
\end{equation}
Throughout this section, we make the following assumption:
\begin{assumption}[Regularity condition]
\[
\sum_{i=1}^\infty {1/\lambda_i} = \infty.
\]
\label{sump:regcond}
\end{assumption}
\noindent
This condition generalizes the classic regularity condition of a univariate birth process \citep{feller1968introduction}.


\begin{thm}
Under Assumption \ref{sump:regcond} (Regularity condition),
\begin{enumerate}[label=(\roman*)]

	\item the forward transition probabilities $\{ P_{\mb{0} \mb{x}}(t) \}_{\mb{x} \leq \mb{B}}$ are the unique solution of the Chapman-Kolmogorov forward equations
\begin{align}
\frac{dP_{\mb{0} \mb{x}}(t)}{dt} =  \sum_{k=1}^d {\lambda^{(k)}_{\mb{x} - \mb{e}_k} P_{\mb{0}, \mb{x} - \mb{e}_k}(t)}  - \left ( \sum_{k=1}^d {\lambda^{(k)}_{\mb{x}}} \right )P_{\mb{0} \mb{x}}(t), \text{ and}
\label{eqn:forward_eq}
\end{align}

	\item the backward transition probabilities $\{ P_{\mb{x} \mb{B}}(t) \}_{\mb{x} \leq \mb{B}}$ are the unique solution of the Chapman-Kolmogorov backward equations
\begin{align}
\frac{dP_{\mb{x} \mb{B}}(t)}{dt} =  \sum_{k=1}^d {\lambda^{(k)}_{\mb{x}} P_{\mb{x} + \mb{e}_k, \mb{B}}(t)}  - \left ( \sum_{k=1}^d {\lambda^{(k)}_{\mb{x}}} \right )P_{\mb{x} \mb{B}
}(t).
\label{eqn:backward_eq}
\end{align}

\end{enumerate}
\label{thm:CKeq}
\end{thm}

\begin{proof}
It is sufficient to prove that the birth rates satisfying Assumption \ref{sump:regcond} uniquely determine the multivariate birth process. By Theorem 7 in \citet{reuter1957}, we have to show that if for some $\zeta >0$, $\{ y_\mb{x} \} \in [0,1]$ satisfies the following equations
\begin{equation}
	\left (\zeta + \sum_{k=1}^d {\lambda^{(k)}_{\mb{x}}} \right ) y_\mb{x}  = \sum_{k=1}^d {\lambda^{(k)}_{\mb{x}} y_{\mb{x} + \mb{e}_k}},
\end{equation}
then $y_\mb{x} = 0$. Let $y_i = \max_{\mb{x} \in D_i} \left \{ y_{\mb{x}} \right \}$ and $\mb{x}^* = \text{argmax}_{\mb{x} \in D_i} \left \{ y_{\mb{x}} \right \}$, we have
\begin{equation}
	\left (\zeta + \sum_{k=1}^d {\lambda^{(k)}_{\mb{x}^*}} \right ) y_i  = \sum_{k=1}^d {\lambda^{(k)}_{\mb{x}^*} y_{\mb{x}^* + \mb{e}_k}} \leq \left ( \sum_{k=1}^d {\lambda^{(k)}_{\mb{x}^*}} \right ) y_{i + 1}.
\end{equation}
Therefore,
\begin{equation}
	\zeta y_i \leq \left ( \sum_{k=1}^d {\lambda^{(k)}_{\mb{x}^*}} \right ) (y_{i + 1} - y_i) \leq \lambda_i (y_{i + 1} - y_i).
	\label{eqn:diffy}
\end{equation}
Assume that there exists $i_0 > 0$ such that $y_{i_0} > 0$. From \eqref{eqn:diffy}, we conclude that for every $i > i_0$, $y_i > y_{i-1}$ and
\begin{equation}
	y_i = \sum_{j = i_0}^{i-1}{\frac{\zeta}{\lambda_j}} + y_{i_0} \to \infty ~~\text{as}~~ i \to \infty,
\end{equation}
which contradicts with $y_i \leq 1$. This contradiction completes the proof.
\end{proof}

Theorem \ref{thm:CKeq} shows that we can evaluate the forward and backward transition probabilities by solving the Chapman-Kolmogorov equations \eqref{eqn:forward_eq} and \eqref{eqn:backward_eq}. However, traditional methods like matrix exponentiation and Euler's method are either computationally expensive or lack numerical accuracy.
Instead, we first solve the Chapman-Kolmogorov equations in the Laplace domain and then apply an inverse Laplace transform to recover $P_{\mb{u} \mb{v}}(t)$.

We define the Laplace transform of $P_{\mb{u} \mb{v}}(t)$ as:
\begin{equation}
f_{\mb{u} \mb{v}}(s) = \mathcal{L}[P_{\mb{u} \mb{v}}(t)](s) = \int_0^\infty{e^{-st}P_{\mb{u} \mb{v}}(t)dt}.
\end{equation}
Note that  $f_{\mb{u} \mb{v}} \ne 0$ if and only if $\mb{u} \leq \mb{v}$.

\begin{cor}
For the multivariate birth process, we have the following recursive formulae:
\begin{equation}
\begin{aligned}
f_{\mb{0} \mb{0}}(s) &= \frac{1}{s +  \sum_{j=1}^d {\lambda^{(j)}_{\mb{0}}}} \\
f_{\mb{B} \mb{B}}(s) &= \frac{1}{s +  \sum_{j=1}^d {\lambda^{(j)}_{\mb{B}}}} \\
f_{\mb{0} \mb{x}}(s) &= \sum_{k=1}^d{\frac{\lambda^{(k)}_{\mb{x} - \mb{e}_k}}{s + \sum_{j=1}^d {\lambda^{(j)}_{\mb{x}}}}} f_{\mb{0}, \mb{x} - \mb{e}_k}(s) \\
f_{\mb{x} \mb{B}}(s) &= \sum_{k=1}^d{\frac{\lambda^{(k)}_{\mb{x}}}{s + \sum_{j=1}^d {\lambda^{(j)}_{\mb{x}}}}} f_{\mb{x} + \mb{e}_k, \mb{B}}(s),
\end{aligned}
\label{eqn:recursive}
\end{equation}
where $\mb{0} \leq \mb{x} \leq \mb{B}$.
\label{cor:recursive}
\end{cor}

\begin{proof}
Applying a Laplace transform to both sides of (\ref{eqn:forward_eq}) and (\ref{eqn:backward_eq}), we arrive at
\begin{equation}
\begin{aligned}
\mathcal{L} \left [ \frac{dP_{\mb{0} \mb{x}}(t)}{dt} \right ](s) &=  \sum_{k=1}^d {\lambda^{(k)}_{\mb{x} - \mb{e}_k} \mathcal{L}[P_{\mb{0}, \mb{x} - \mb{e}_k}(t)](s)}  - \left ( \sum_{k=1}^d {\lambda^{(k)}_{\mb{x}}} \right ) \mathcal{L}[P_{\mb{0} \mb{x}}(t)](s) \text{ and} \\
\mathcal{L} \left [ \frac{dP_{\mb{x} \mb{B}}(t)}{dt} \right ](s) &=  \sum_{k=1}^d {\lambda^{(k)}_{\mb{x}} \mathcal{L}[P_{\mb{x} + \mb{e}_k, \mb{B}}(t)](s)}  - \left ( \sum_{k=1}^d {\lambda^{(k)}_{\mb{x}}} \right ) \mathcal{L}[P_{\mb{x} \mb{B}
}(t)](s).
\end{aligned}
\end{equation}
Noting that
\begin{equation}
\begin{aligned}
\mathcal{L} \left [ \frac{dP_{\mb{0} \mb{x}}(t)}{dt} \right ](s) &= s \mathcal{L}[P_{\mb{0} \mb{x}}(t)](s) - P_{\mb{0} \mb{x}}(0) \text{ and} \\
\mathcal{L} \left [ \frac{dP_{\mb{x} \mb{B}}(t)}{dt} \right ](s) &= s \mathcal{L}[P_{\mb{x} \mb{B}}(t)](s) - P_{\mb{x} \mb{B}}(0)
\end{aligned}
\end{equation}
enables us to write
\begin{equation}
\begin{aligned}
s f_{\mb{0} \mb{x}}(s) - P_{\mb{0} \mb{x}}(0) &= \sum_{k=1}^d {\lambda^{(k)}_{\mb{x} - \mb{e}_k} f_{\mb{0}, \mb{x} - \mb{e}_k}(s)}  - \left ( \sum_{k=1}^d {\lambda^{(k)}_{\mb{x}}} \right )f_{\mb{0} \mb{x}}(s) \text{ and} \\
s f_{\mb{x} \mb{B}}(s) - P_{\mb{x} \mb{B}}(0) & = \sum_{k=1}^d {\lambda^{(k)}_{\mb{x}} f_{\mb{x} + \mb{e}_k, \mb{B}}(s)}  - \left ( \sum_{k=1}^d {\lambda^{(k)}_{\mb{x}}} \right )f_{\mb{x} \mb{B}}(s).
\label{eqn:Lapdomain}
\end{aligned}
\end{equation}
From \eqref{eqn:Lapdomain}, we have
\begin{equation}
\begin{aligned}
s f_{\mb{0} \mb{0}}(s) - P_{\mb{0} \mb{0}}(0) &= \sum_{k=1}^d {\lambda^{(k)}_ {- \mb{e}_k} f_{\mb{0}, - \mb{e}_k}(s)}  - \left ( \sum_{k=1}^d {\lambda^{(k)}_{\mb{0}}} \right )f_{\mb{0} \mb{0}}(s) \text{ and} \\
s f_{\mb{B} \mb{B}}(s) - P_{\mb{B} \mb{B}}(0) & = \sum_{k=1}^d {\lambda^{(k)}_{\mb{B}} f_{\mb{B} + \mb{e}_k, \mb{B}}(s)}  - \left ( \sum_{k=1}^d {\lambda^{(k)}_{\mb{B}}} \right )f_{\mb{B} \mb{B}}(s).
\end{aligned}
\end{equation}
Since $P_{\mb{0} \mb{0}}(0)  = P_{\mb{B} \mb{B}}(0) = 1$ and $P_{\mb{0}, - \mb{e}_k}(t)  = P_{\mb{B} + \mb{e}_k, \mb{B}}(t) = 0$, we deduce
\begin{equation}
\begin{aligned}
f_{\mb{0} \mb{0}}(s) &= \frac{1}{s +  \sum_{j=1}^d {\lambda^{(j)}_{\mb{0}}}} \text{ and} \\
f_{\mb{B} \mb{B}}(s) &= \frac{1}{s +  \sum_{j=1}^d {\lambda^{(j)}_{\mb{B}}}}.
\end{aligned}
\end{equation}
Moreover, $P_{\mb{0} \mb{x}}(0) = 0$ for $\mb{x} \ne \mb{0}$ and $P_{\mb{x} \mb{B}}(0) = 0$ for $\mb{x} \ne \mb{B}$.
Hence, from \eqref{eqn:Lapdomain}, we obtain
\begin{equation}
\begin{aligned}
f_{\mb{0} \mb{x}}(s) &= \sum_{k=1}^d{\frac{\lambda^{(k)}_{\mb{x} - \mb{e}_k}}{s + \sum_{j=1}^d {\lambda^{(j)}_{\mb{x}}}}} f_{\mb{0}, \mb{x} - \mb{e}_k}(s) \text{ and} \\
f_{\mb{x} \mb{B}}(s) &= \sum_{k=1}^d{\frac{\lambda^{(k)}_{\mb{x}}}{s + \sum_{j=1}^d {\lambda^{(j)}_{\mb{x}}}}} f_{\mb{x} + \mb{e}_k, \mb{B}}(s).
\end{aligned}
\end{equation}
Thus, the proof is completed.
\end{proof}
From Corollary \ref{cor:recursive}, we can derive analytic formulae for all $\{ f_{\mb{0} \mb{x}}(s) \}_{\mb{x} \leq \mb{B}}$ and $\{ f_{\mb{x} \mb{B}}(s) \}_{\mb{x} \leq \mb{B}}$.
For $\mb{u} \leq \mb{v}$, let a path from $\mb{u}$ to $\mb{v}$ be an increasing sequence $\mb{p} = \{\mb{p}_i\}_{i=1}^{n}$ such that
\[
\mb{p}_1 = \mb{u},~~\mb{p}_n = \mb{v},~~ \mb{p}_i \leq  \mb{p}_{i+1} ,~~\text{and}~~ \mb{p}_{i+1} - \mb{p}_i \in \{ \mb{e}_1, \mb{e}_2, \ldots, \mb{e}_d\}.
\]
Denote $\mc{P}_{\mb{u} \mb{v}}$ and $\mc{I}_i$ to be the set of all paths from $\mb{u}$ to $\mb{v}$ and the index of the only non-zero coordinate of $\mb{p}_{i+1} - \mb{p}_i$ respectively. We have
\begin{equation}
\begin{aligned}
    f_{\mb{0} \mb{x}}(s) &= \frac{1}{s +  \sum_{j=1}^d {\lambda^{(j)}_{\mb{0}}}} \left ( \sum_{\mb{p} \in \mc{P}_{\mb{0} \mb{x}}}{ \prod_{i = 2}^n{\frac{\lambda^{(\mc{I}_{i-1})}_{\mb{p}_{i-1}}}{s + \sum_{j=1}^d {\lambda^{(j)}_{\mb{p}_i}}}}} \right ) \\
f_{\mb{x} \mb{B}}(s) &= \frac{1}{s +  \sum_{j=1}^d {\lambda^{(j)}_{\mb{B}}}} \left ( \sum_{\mb{p} \in \mc{P}_{\mb{x} \mb{B}}}{ \prod_{i = 1}^{n-1}{\frac{\lambda^{(\mc{I}_i)}_{\mb{p}_{i}}}{s + \sum_{j=1}^d {\lambda^{(j)}_{\mb{p}_i}}}}} \right ).
\end{aligned}
\label{eqn:analytic}
\end{equation}
However, evaluating $\{ f_{\mb{0} \mb{x}}(s) \}_{\mb{x} \leq \mb{B}}$ and $\{ f_{\mb{x} \mb{B}}(s) \}_{\mb{x} \leq \mb{B}}$ using \eqref{eqn:analytic} is infeasible because the number of paths from $\mb{0}$ to $\mb{B}$ is extremely large. For example, when all the birth rates are positive, the number of paths is
\begin{equation}
\prod_{i = 1}^{d}{\frac{\left (\sum_{j=i}^d{B_j} \right ) !}{B_i ! \left (\sum_{j=i+1}^d{B_j} \right ) !}} .
\label{eqn:numpath}
\end{equation}
For example, when $d = 2$ and $B_1 = B_2 = B$, the number of paths \eqref{eqn:numpath} becomes $(B+1) (B+ 2) \cdots (2B) > B^B$.

\par

The sum-product structure in \eqref{eqn:analytic} suggests that dynamic programming may lead to efficient computation of $\{ f_{\mb{0} \mb{x}}(s) \}_{\mb{x} \leq \mb{B}}$ and $\{ f_{\mb{x} \mb{B}}(s) \}_{\mb{x} \leq \mb{B}}$ that we achieve through the recursive formulae \eqref{eqn:recursive}.
The computation cost of the recursion is only $\mathcal{O}(\prod_{k=1}^d{B_k})$ because we need one loop for each coordinate.
Algorithm \ref{alg:dynamic} presents pseudo-code for computing $\{ f_{\mb{0} \mb{x}}(s) \}_{\mb{x} \leq \mb{B}}$ via dynamic programming.
The algorithm for evaluating $\{ f_{\mb{x} \mb{B}}(s) \}_{\mb{x} \leq \mb{B}}$ is similar.

\begin{algorithm}
\begin{algorithmic}[1]
\REQUIRE $s > 0$, $\{\lambda_{\mb{x}}^{(j)}\}_{j=1}^d$
\STATE $f_{\mb{0} \mb{0}} \leftarrow 1$
\FOR{$i_1 = 0$ to $B_1$}
\FOR{$i_2 = 0$ to $B_2$}
\STATE $\vdots$
\FOR{$i_d = 0$ to $B_d$}
\STATE $\mb{x} \leftarrow (i_1, i_2, \ldots, i_d)$
\STATE $m \leftarrow s + \sum_{j=1}^d{\lambda_{\mb{x}}^{(j)}}$
\STATE $f_{\mb{0} \mb{x}} \leftarrow f_{\mb{0} \mb{x}}/m$
\FOR{$k = 1$ to $d$}
\IF{$i_k < B_k$}
\STATE $f_{\mb{0}, \mb{x}+\mb{e}_k} \leftarrow f_{\mb{0}, \mb{x}+\mb{e}_k} + \lambda_{\mb{x}}^{(k)} \times f_{\mb{0} \mb{x}}$
\ENDIF
\ENDFOR
\ENDFOR
\STATE $\vdots$
\ENDFOR
\ENDFOR
\end{algorithmic}
\caption{Dynamic programming algorithm for computing $\{ f_{\mb{0} \mb{x}}(s) \}_{\mb{x} \leq \mb{B}}$.}
 \label{alg:dynamic}
\end{algorithm}

Then, we approximate the inverse Laplace transform of $f_{\mb{u} \mb{v}}(s)$ by the method proposed in \citet[][equation (4.6)]{abate1992}:
 \begin{equation}
 P_{\mb{u} \mb{v}}(t) = \mathcal{L}^{-1}(f_{\mb{u} \mb{v}})(t) \approx \frac{e^{M/2}}{2t} {\cal R} \left [ f_{\mb{u} \mb{v}} \left ( \frac{M}{2t} \right ) \right ] + \frac{e^{M/2}}{t} \sum_{k=1}^\infty{(-1)^k {\cal R} \left [ f_{\mb{u} \mb{v}} \left ( \frac{M + 2k \pi i}{2t} \right ) \right ]} ,
 \label{eqn:invLap}
 \end{equation}
where ${\cal R} [z]$ is the real part of $z$.
Here, the positive number $M$ is used to control the discretization error.
Specifically, the discretization error is
\[
\sum_{k=1}^\infty e^{-k M}  P_{\mb{u} \mb{v}}((2k + 1)t),
\]
which can be bounded by $1/(e^M - 1)$.
However, \citet{abate1992} warn that we should not choose $M$ too large because it makes the infinite sum $(\ref{eqn:invLap})$ harder to evaluate.
They suggest to aim for $10^{-7}$ to $10^{-8}$ accuracy on a machine with $14$-digit precision.
Follow this instruction, we choose $M = 20$ throughout this paper.
We opt to use a Levin acceleration method \citep{levin1973} to improve the convergence rate of $(\ref{eqn:invLap})$.
Let $L$ be the number of iterations required from Levin acceleration to achieve a certain error bound for the approximation $(\ref{eqn:invLap})$, then we have the following corollary:

\begin{cor}
The total complexity of our algorithm to compute $\{ P_{\mb{0} \mb{x}}(t) \}_{\mb{x} \leq \mb{B}}$ and $\{ P_{\mb{x} \mb{B}}(t) \}_{\mb{x} \leq \mb{B}}$ is $\mathcal{O}(L \prod_{k=1}^d{B_k})$.
\label{cor:MBcomplex}
\end{cor}

Note that when we aim for $10^{-8}$ accuracy, $L$ usually ranges from $100$ to $1000$.


\subsection{Re-parameterization}

Given an $m$-compartmental process $\mb{Y}(t)$ with $d$ possible types of transition between compartments,
computing the transition probability $\Pr \{\mb{Y}(t)= \mb{v}~ |~\mb{Y}(0)=\mb{u}\}$ by solving the compartmental Chapman-Kolmogorov equations is generally intractable because, unlike multivariate birth processes, individual compartment population $Y_i(t)$ may increase or decrease over time.
Here, we recast $\mb{Y}(t)$ into a $d$-dimensional birth process $\mb{X}(t)$ and aim to compute the transition probabilities of $\mb{Y}(t)$ from the transition probabilities of $\mb{X}(t)$.

We denote $i \to j$ be a transition from compartment $\comp_i$ to compartment $\comp_j$.
For $k = 1, 2 , \ldots, d$, let $i_k \to j_k$ be the $k$-th type of transition.
We construct $\mb{X}(t)$ by letting $X_k(t)$ be the number of $k$-type transition events happening from time $0$ to $t$.
Define an $m \times d$ matrix $\map = [a_{lk}]$ as follows:
\begin{equation}
a_{lk} =
\left\{
\begin{array}{rl}
-1, & \mbox{if } l = i_k \\
1, & \mbox{if } l = j_k \\
0, & \mbox{otherwise} ,
\end{array}
\right.
\end{equation}
then we have the following lemma:

\begin{lemma}
$\mb{Y}(t) =\mb{Y}(0) + [\map \mb{X}(t)]^T$ where $T$ denotes the matrix transpose.
Moreover, the birth rates for $\mb{X}(t)$ are $\lambda^{(k)}_\mb{x} = \mu_{i_k j_k}(\theta, \mb{Y}(0) + [\map \mb{x}]^T)$.
\label{lem:identity}
\end{lemma}

Define $W = \{ \mb{w} \in \mbb{N}^d:   \map \mb{w} = (\mb{u} - \mb{v})^T \}$. By Lemma \ref{lem:identity}, we deduce that
\begin{equation}
Pr \{\mb{Y}(t)= \mb{v}~ |~\mb{Y}(0)=\mb{u}\} = \sum_{\mb{w} \in W}{Pr \{\mb{X}(t)= \mb{w}~ |~\mb{X}(0)=\mb{0}\}}.
\label{eqn:YtoX}
\end{equation}
We want to employ Equation \eqref{eqn:YtoX} for computing the transition probabilities of $\mb{Y}(t)$.
However, evaluating the summation in \eqref{eqn:YtoX} is infeasible when the set $W$ has infinitely many elements.
To limit the cardinality of $W$, we proffer a small restriction on the class of compartmental models for which we can compute their transition probabilities in polynomial complexity.

\begin{assumption}[Finite loops]
Each individual visits each compartment at most $\upper$ times between two consecutive observations.
\label{sump:visit}
\end{assumption}

Assumption \ref{sump:visit} is rarely restrictive for many compartmental models for infectious diseases.
Infected individuals usually develop at least partial immunity to re-inflection that wanes at a rate commensurate with or slower than the observation process.
Further, it is notable that if a compartmental model can be represented by a directed acyclic graph, then an individual never returns to a compartment after leaving.
In this case, this assumption is satisfied with $\upper = 1$.

\begin{thm}
For a compartmental model satisfying Assumption \ref{sump:visit} (Finite loops), the complexity for computing its transition probabilities via Equation \eqref{eqn:YtoX} is $\mathcal{O}(L \upper^d N^d)$, where $N$ is the total population of all compartments.
\label{thm:complexity}
\end{thm}

\begin{proof}
By Assumption \ref{sump:visit}, $\mb{X} \leq \upper \mb{N}$ where $\mb{N}$ is a $d$-dimensional vector $(N_1, N_2, \ldots, N_d)$.
Hence $\lambda^{(k)}_\mb{x} = 0$ when $\mb{x} \geq \upper \mb{N}$.
By Theorem \ref{thm:CKeq} and Corollary \ref{cor:MBcomplex}, we can compute the transition probabilities $(Pr \{\mb{X}(t)= \mb{x}~ |~\mb{X}(0)=\mb{0}\})_{\mb{x} \leq \upper \mb{N}}$ at a cost of $\mathcal{O}(L \upper^d N^d)$.
Then, we can compute the transition probabilities of $\mb{Y}$ through Equation \eqref{eqn:YtoX} at the same cost.
Therefore, the total complexity is $\mathcal{O}(L \upper^d N^d)$.
\end{proof}


\section{Compartmental models of infectious diseases}
\label{sec:epimodels}

We apply our recursion method to three prevailing compartmental models of infectious diseases including the SIR, SEIR and SIRS models.


\subsection{Susceptible-infectious-removed model}
Proposed by \citet{mckendrick1926applications}, the stochastic SIR model is probably the most famous compartmental model in epidemiology.
This model divides the population into three different compartments: susceptible ($S$), infectious ($I$), and removed ($R$), and allows two possible transitions: infection ($S \to I$) with rate $\beta S I$ and removal ($I \to R$) with rate $\gamma I$.
Here, $\beta > 0$ is the infection rate and $\gamma > 0$ is the removal rate of the disease.
Figure \ref{fig:SIR} visualizes the directed graph representing this model.

\begin{figure}[h]
\centering
\begin{tikzpicture}
\shadedraw[ball color=black!50, opacity=0.4] (0, 0) circle (0.8) node[opacity=1.0] {\Huge S};
\shadedraw[ball color=black!50, opacity=0.4] (3, 0) circle (0.8) node[opacity=1.0] {\Huge I};
\shadedraw[ball color=black!50, opacity=0.4] (6, 0) circle (0.8) node[opacity=1.0] {\Huge R};
%
%
\draw[->, line width=2pt] (1, 0) -- (1.5, 0) node[above] {{\Large $\beta S I $}} -- (2, 0);
\draw[->, line width=2pt] (4, 0) -- (4.5, 0) node[above] {{\Large $\gamma I$}} -- (5, 0);
\end{tikzpicture}
\caption{A directed graph representation of the SIR model.}
\label{fig:SIR}
\end{figure}
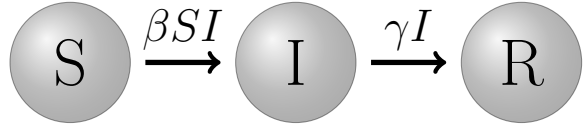

Because the total population $S(t) + I(t) + R(t)$ is constant, \citet{ho2016} consider $\{S(t), I(t)\}$ as a death/birth-death process and propose an algorithm to compute its transition probabilities using a continued fraction representation.
The computational cost of this algorithm for evaluating the full transition probability matrix is $\mc{O}(LN^3)$.
Our present method re-parameterizes the SIR model using number of infection events $N_{SI}(t)$ and removal events $N_{IR}(t)$.
Note that there is a one-to-one correspondence between $\{S(t), I(t)\}$ and $\{ N_{SI}(t), N_{IR}(t) \}$:
\begin{equation}
 \begin{pmatrix} S(t) \\  I(t) \end{pmatrix} =
 \begin{pmatrix} s_0 \\  i_0 \end{pmatrix} +
 \begin{pmatrix} -1 & 0 \\ 1 & -1 \end{pmatrix}
  \begin{pmatrix} N_{SI}(t) \\ N_{IR}(t) \end{pmatrix} ,
\end{equation}
where $(s_0, i_0)$ is the realized value of $\{S(0), I(0)\}$.
It follows that $\{ N_{SI}(t), N_{IR}(t) \}$ is a bivariate birth process with birth rates $\beta (s_0 - N_{SI})^+ (i_0 + N_{SI} - N_{IR})^+$ and $\gamma (i_0 + N_{SI} - N_{IR})^+$ where $a^+ = \max \{ a, 0 \}$.
Since the directed graph representing the SIR model is acyclic, Assumption \ref{sump:visit} is satisfied with $\upper = 1$.
By Theorem \ref{thm:complexity}, we have the following Corollary:

\begin{cor}
The complexity for evaluating the full transition probability matrix of the SIR model using our method is $\mc{O}(LN^2)$.
\end{cor}

We remark that our present method is an order of magnitude in $N$ faster than that of \citet{ho2016} for computing the entire transition probability matrix.
In practice, however, we often only need to compute the transition probabilities between observations.
In this case, the computational cost of our present method decreases further to $\mc{O}(L \Delta_S \Delta_R)$ where $\Delta_S$ and $\Delta_R$ are the changes in susceptible and removed populations between observations.
In many situations, $\Delta_S$ and $\Delta_R$ are significantly smaller than the total population $N$, for example, when tracing the dynamics of a rare disease across an entire nation.
We implement our method in the \texttt{R}  function \texttt{SIR\_prob} (\texttt{MultiBD} package) \url{https://github.com/msuchard/MultiBD}.

\begin{figure}[h]
\begin{center}
\includegraphics[width=0.5\textwidth]{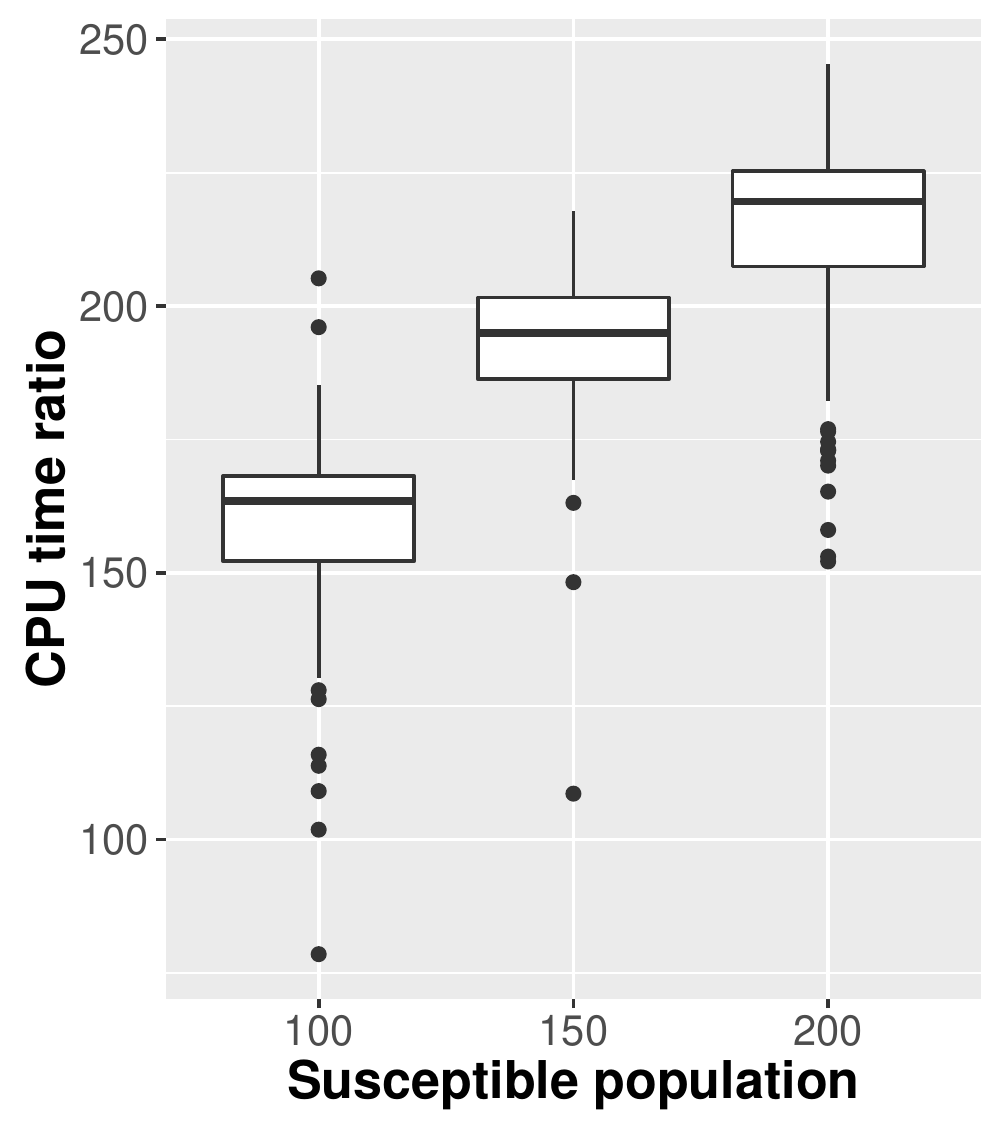}
\end{center}
\caption{CPU time ratios of the continued fraction method (\texttt{dbd\_prob}) to the proposed recursion method (\texttt{SIR\_prob}) for computing the full transition probabilities matrix of the SIR model with $\gamma = 2.73$ and $\beta = 0.0178$. We set $I(0) = 1$ and $S(0) = 100, 150, 200$.}
\label{fig:simulation_SIR}
\end{figure}

To illustrate the computation gain of our recursion method compared to the continued fraction representation of \citet{ho2016}, we evaluate the full forward transition probability matrix of the SIR model with $\gamma = 2.73$ and $\beta = 0.0178$ \citep[estimated values from the Eyam plague data by][]{ho2016} using both methods.
The death/birth-death method in \citet{ho2016} is implemented in the \texttt{R} function \texttt{dbd\_prob} (\texttt{MultiBD} package).
We set the starting infectious population $i_0$ to be $1$ and consider 3 different starting susceptible populations $s_0 = 100, 150, 200$.
For each scenario, we repeat the evaluation a hundred times and compare the computing times and the results from both methods.
Figure \ref{fig:simulation_SIR} summarizes this comparison, and we see that \texttt{SIR\_prob} is more than $150$ times faster than \texttt{dbd\_prob}.
On the other hand, the two methods return similar transition probability matrices whose $L_1$ distance is less than $10^{-12}$. Here, the $L_1$ distance between two matrices $\mb{A} = (a_{ij})$ and $\mb{B}=(b_{ij})$ is $\sum_{ij}{|a_{ij} - b_{ij}|}$.


\subsection{Susceptible-exposed-infectious-removed model}
The SEIR model extends the SIR model by adding an exposed (E) compartment.
We visualize the SEIR model by the directed acyclic graph in Figure \ref{fig:SEIR}.

\begin{figure}[h]
\centering
\begin{tikzpicture}
%
%
\shadedraw[ball color=black!50, opacity=0.4] (0, 0) circle (0.8) node[opacity=1.0] {\Huge S};
\shadedraw[ball color=black!50, opacity=0.4] (3, 0) circle (0.8) node[opacity=1.0] {\Huge E};
\shadedraw[ball color=black!50, opacity=0.4] (6, 0) circle (0.8) node[opacity=1.0] {\Huge I};
\shadedraw[ball color=black!50, opacity=0.4] (9, 0) circle (0.8) node[opacity=1.0] {\Huge R};
%
%
\draw[->, line width=2pt] (1, 0) -- (1.5, 0) node[above] {{\Large $\beta S I $}} -- (2, 0);
\draw[->, line width=2pt] (4, 0) -- (4.5, 0) node[above=3pt] {{\Large $\kappa E$}} -- (5, 0);
\draw[->, line width=2pt] (7, 0) -- (7.5, 0) node[above] {{\Large $\gamma I$}} -- (8, 0);
\end{tikzpicture}
\caption{A directed graph representation of the SEIR model.}
\label{fig:SEIR}
\end{figure}
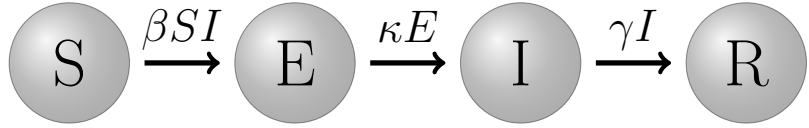

Let $\{ N_{SE}(t), N_{EI}(t), N_{IR}(t) \}$ be the number of transition events $S \to E$, $E \to I$, and $I \to R$ respectively.
Then, we have an one-to-one correspondence with $\{S(t), E(t), I(t)\}$ as follows:
\begin{equation}
 \begin{pmatrix} S(t) \\  E(t) \\ I(t) \end{pmatrix} =
 \begin{pmatrix} s_0 \\  e_0 \\ i_0 \end{pmatrix} +
 \begin{pmatrix} -1 & 0 & 0 \\ 1 & -1 & 0 \\ 0 & 1 & -1 \end{pmatrix}
  \begin{pmatrix} N_{SE}(t) \\ N_{EI}(t) \\  N_{IR}(t) \end{pmatrix} ,
\end{equation}
where $(s_0, e_0, i_0)$ is the realized value of $\{S(0), E(0), I(0)\}$.
Again, $\{ N_{SE}, N_{EI}, N_{IR}\}$ is a trivariate birth process with birth rates $ \beta (s_0 - N_{SE})^+ (i_0 + N_{EI} - N_{IR})^+$, $\kappa (e_0 + N_{SE} - N_{EI})^+$, and $\gamma (i_0 + N_{EI} - N_{IR})^+$.
By Theorem \ref{thm:complexity}, we have:

\begin{cor}
The complexity for evaluating the full transition probability matrix of the SEIR model using our method is $\mc{O}(LN^3)$.
\end{cor}


\subsection{Susceptible-infectious-removed-susceptible model}

For some diseases, removed persons can lose immunity, making possible transition from the ``recovered'' (R) to ``susceptible'' (S) compartments.
The SIRS model takes into account these scenarios by allowing the transition $R \to S$.
Figure \ref{fig:SIRS} visualizes the directed graph representing the SIRS model.

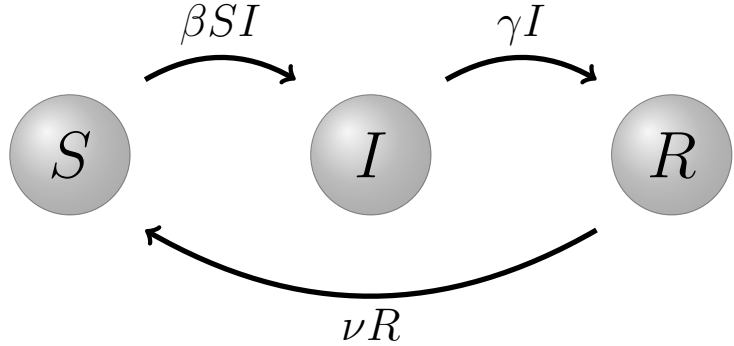
\begin{figure}[h]
\centering
\begin{tikzpicture}
%
%
\shadedraw[ball color=black!50, opacity=0.4] (0, 0) circle (0.8) node[opacity=1.0] {\Huge $S$};
\shadedraw[ball color=black!50, opacity=0.4] (4, 0) circle (0.8) node[opacity=1.0] {\Huge $I$};
\shadedraw[ball color=black!50, opacity=0.4] (8, 0) circle (0.8) node[opacity=1.0] {\Huge $R$};
%
%
\draw[bend left,->, line width=2pt] (1, 1) to node[above] {{\Large $\beta S I $}} (3, 1);
\draw[bend left, ->, line width=2pt] (5, 1) to node[above] {{\Large $\gamma I$}} (7, 1);
\draw[bend left, ->, line width=2pt] (7, -1) to node[below] {{\Large $\nu R$}} (1, -1);
\end{tikzpicture}
\caption{A directed graph representation of the SIRS model.}
\label{fig:SIRS}
\end{figure}

Denote $\{ N_{SI}(t), N_{IR}(t), N_{RS}(t) \}$ as the number of transition events $S \to I$, $I \to R$, and $R \to S$ respectively.
We have
\begin{equation}
 \begin{pmatrix} S(t) \\  I(t) \\ R(t) \end{pmatrix} =
 \begin{pmatrix} s_0 \\  i_0 \\ r_0 \end{pmatrix} +
 \begin{pmatrix} -1 & 0 & 1 \\ 1 & -1 & 0 \\ 0 & 1 & -1 \end{pmatrix}
  \begin{pmatrix} N_{SI}(t) \\ N_{IR}(t) \\  N_{RS}(t) \end{pmatrix} ,
\end{equation}
where $(s_0, i_0, r_0)$ is the realized value of $\{S(0), I(0), R(0)\}$.
%
In this situation, $\{ N_{SI}, N_{IR}, N_{RS}\}$ is a trivariate birth process with birth rates $\beta (s_0 + N_{RS} - N_{SI})^+ (i_0 + N_{SI} - n_{IR})^+$, $\gamma (i_0 + n_{SI} - n_{IR})^+$, and $\nu (r_0 + n_{IR} - n_{RS})^+$.
In practice, $\nu$ is much smaller than $\beta \times I(t)$ and $\gamma$. 
Hence, we can assume that during $(0,t)$ each individual can only be infected at most $\upper$ times.
By Theorem \ref{thm:complexity}, we arrive at
\begin{cor}
The complexity for evaluating the full transition probability matrix of the SIRS model using our method is $\mc{O}(L \upper^3 N^3)$.
\end{cor}

\subsection{Comparisons}

\newcommand{\idx}{m}

We use prevalence counts from the plague in Eyam from June 18th to October 20th, 1666 \citep{raggett1982} to compare our recursion method with the SMC algorithm implemented in the \texttt{R} function \texttt{pfilter} \citep[\texttt{pomp} package][]{king2015} and the matrix exponentiation method implemented in the state-of-the-art software \texttt{Expokit} \citep{sidje1998}.
Plague is a deadly infectious disease caused by the bacterium \textit{Yersinia pestis}. It is mainly spread by infected fleas from small animals, particularly rodents, and has killed 100s of millions of people through human history.
In Eyam, only 83 of the original 350 villagers survived at the end of the plague.
The data contain the susceptible and infectious populations $\{ (s_{\idx}, i_{\idx}) \}_{{\idx}=1}^n$ in Eyam at time $\{ t_{\idx} \}_{{\idx}=1}^n$.
The log likelihood function is
\begin{equation}
\log l(\beta, \gamma | \{ (s_{\idx}, i_{\idx}) \}_{{\idx}=1}^n) = \sum_{{\idx}=1}^{n-1}{ \log \Pr \left \{ \begin{array}{c | c} S(t_{{\idx}+1}) = s_{{\idx}+1}  & S(t_{\idx}) = s_{\idx} \\ I(t_{{\idx}+1}) = i_{{\idx}+1} & I(t_{\idx}) = i_{\idx} \end{array} \right \}}.
\label{eqn:loglik}
\end{equation}
We compute the log likelihood \eqref{eqn:loglik} under the stochastic SIR model with $\beta = 0.0178$ and $\gamma = 2.73$ \citep[estimated values from the Eyam plague data by][]{ho2016}.

\subsubsection{Comparing to sequential Monte Carlo}

The likelihood calculation is repeated a thousand times and the number of attempted simulant particles for each estimation for \texttt{pfilter} is set as $1000, 2000, 3000$, and $4000$.
For these data and parameter estimates,
\texttt{pfilter} fails to achieve a 100\% success rate for approximating the likelihood.
The success rate is low with $1000$ particles (only 20.1\%), and increases as the number of particles increases (see Table \ref{tab:pomp}).
Filtering failure occurs when all particles become incompatible with the data counts; this can happen frequently when the counts are observed without error.
When filtering succeeds, the approximation is fairly similar to our method, and the standard deviation of these approximations, while sizable, decreases from $1.28$ to $0.97$ as the number of particles increases.
When filtering fails, the approximation is off target by a large margin.
The computation time of \texttt{pfilter} is about $10$ times slower compared to our algorithm for every $1000$ particles (Table \ref{tab:pomp}).
This comparison shows that our recursion method is faster
than the SMC method.
Moreover, our method is stable while approximations using SMC are very unstable due to a high failure rate.
It is worth mentioning that SMC is known to be an inefficient algorithm for computing the likelihood when the observations have no error.

\begin{table}[h]
\begin{center}
\begin{tabular}{ lcccc }
\hline
Number of particles & 1000 & 2000 & 3000 & 4000 \\
\hline
\hline
Success rate & 20.1\% & 53.2\% & 71.1\% & 78.8\% \\
Average time ratio & 10.14 & 20.11 & 30.09 & 40.1 \\
Standard deviation & 1.28 & 1.11 & 1.06 & 0.97 \\
\hline
\end{tabular}
\end{center}
\caption{Success rates of sequential Monte Carlo method (\texttt{pfilter}) and its average computing time ratios compared to our algorithm.}
\label{tab:pomp}
\end{table}


\subsubsection{Comparing to matrix exponentiation method}

To evaluate the log likelihood \eqref{eqn:loglik} via matrix exponentiation, we use the function \texttt{expv} in \texttt{expoRkit}, an \texttt{R}-interface to the Fortran package \texttt{Expokit}, to compute the transition probabilities.
Again, the likelihood calculation is repeated a thousand times.
Our method and matrix exponentiation method produce similar results: the difference is less than $1.53 \times 10^{-7}$.
In term of speed, the average CPU computation time ratio of matrix exponentiation method to our method is $15$ and the standard deviation is $1$.
Therefore, our method is more efficient in computing the likelihood function of the stochastic SIR model than matrix exponentiation method.


\section{Further statistical applications}
\label{sec:addapp}


The ability to efficiently compute the likelihood function makes it straightforward to use maximum likelihood estimators and Metropolis-Hasting algorithms for Bayesian inference.
In this section, we provide two additional extensions that the recursion opens up to us
that were unavailable with previous methods.
The first application is inference via HMC, which requires evaluating the derivative of the posterior distribution with respect to the unknown model parameters.
The second application is accessing model adequacy for the classic SIR model using Bayes factors.


\subsection{Inference via Hamiltonian Monte Carlo}

HMC is a MCMC method using Hamiltonian dynamics to produce proposals for sampling from a continuous distribution on $\mathbb{R}^d$. Hamiltonian dynamics contain ``location" variables $q$, that are the parameters of interest, and nuisance ``momentum" variables $p$ \citep[see][for an excellent review]{neal2011mcmc}.
In a Bayesian setting, we may treat the negative log of the posterior distribution as the potential energy function:
\begin{equation}
U(q) = - \log \left [  l(q | \mb{D}) \pi(q) \right ] ,
\end{equation}
where $l(q | \mb{D})$ is the likelihood given data $\mb{D}$ and $\pi(q)$ is the prior distribution. On the other hand, researchers often place a multivariate Normal distribution $\mc{N}(0, \Sigma)$ on $p$ and let $p$ be independent of $q$.
Typically, $\Sigma$ is the identity matrix and the corresponding kinetic energy function is
\[
K(p) = \sum_{i=1}^d{\frac{p_i^2}{2}}.
\]
The Hamiltonian is defined as $H(q,p) = U(q) + K(p)$, and the Hamiltonian dynamics follow the following system of partial differential equations:
\begin{equation}
\begin{aligned}
\frac{dq_i}{dt} &= \frac{\partial H}{\partial p_i} = p_i \\
\frac{dp_i}{dt} &= - \frac{\partial H}{\partial q_i} = - \frac{\partial U}{\partial q_i}.
\end{aligned}
\label{eqn:Hamilton}
\end{equation}
The HMC algorithm consists two steps. In the first step, a proposal for $p$ is sampled from $\mc{N}(0, \Sigma)$. In the second step, $(q_t, p_t)$ is obtained from the Hamiltonian dynamics \eqref{eqn:Hamilton} starting at the current value $(q_c,p_c)$. In practice, we may use a leapfrog integration scheme to approximate the solution of \eqref{eqn:Hamilton}. The proposal $(q^*, p^*)$ is set as $(q_t, - p_t)$ and is accepted with probability
\begin{equation}
\min \left [  1, e^{H(q_c, p_c) - H(q^*, p^*)} \right ].
\end{equation}

The ability to efficiently compute the derivatives of the transition probabilities with respect to $q$ opens the possibility of using HMC for studying infectious disease epidemics.
To illustrate, we employ HMC to analyze the $17^{\mbox{\tiny th}}$ century plague in Eyam.
Denote
\begin{equation}
P_{\idx} = \Pr \left \{ \begin{array}{c | c} S(t_{{\idx}+1}) = s_{{\idx}+1}  & S(t_{\idx}) = s_{\idx} \\ I(t_{{\idx}+1}) = i_{{\idx}+1} & I(t_{\idx}) = i_{\idx} \end{array} \right \}.
\end{equation}
Then, the log likelihood function \eqref{eqn:loglik} can be written as
\begin{equation}
\log l(\beta, \gamma | \{ (s_{\idx}, i_{\idx}) \}_{{\idx}=1}^n) = \sum_{{\idx}=1}^{n-1}{ \log P_{\idx}}.
\end{equation}
To satisfy positivity constraints, we opt to use $(u,v) := (\log \beta, \log \gamma)$ as our parameters instead of $(\beta, \gamma)$.
To apply HMC, we derive the derivatives of $\log l$ with respect to $u$ and $v$:
\begin{equation}
\begin{aligned}
\frac{\partial \log l}{\partial u} &= \frac{\partial \log l}{\partial \beta} \frac{\partial \beta}{\partial u} = \sum_{{\idx}=1}^{n-1}{ \frac{P_{\idx}^{(\beta)}}{P_{\idx}}} \beta \\
\frac{\partial \log l}{\partial v} &= \frac{\partial \log l}{\partial \gamma} \frac{\partial \gamma}{\partial v} = \sum_{{\idx}=1}^{n-1}{ \frac{P_{\idx}^{(\gamma)}}{P_{\idx}}} \gamma.
\end{aligned}
\end{equation}
We assume \textit{a priori} that $u \sim \mc{N}(0, 100^2)$ and $v \sim \mc{N}(0, 100^2)$. We explore the posterior distribution of $(u,v)$ using HMC with
$10000$ iterations and discard the first $2000$ iterations.
Figure \ref{fig:HMC} visualizes the posterior density of $(u,v)$.
This result is similar to the density estimation using a Metropolis-Hasting algorithm performed in \citet{ho2016},
but at a significant time cost savings.
The average effective sample size per unit-time of HMC is 60-fold larger, mostly owing to substantial computational order reduction in the likelihood evaluation under our multivariate-birth process formulation.
The posterior means of $\beta$ and $\gamma$ are $0.0197$ and $3.22$. The $95\%$ Bayesian credible intervals are $(0.0164, 0.0234)$ and $(2.69, 3.83)$ respectively.





\begin{figure}[!h]
\begin{center}
\includegraphics[width=0.5\textwidth]{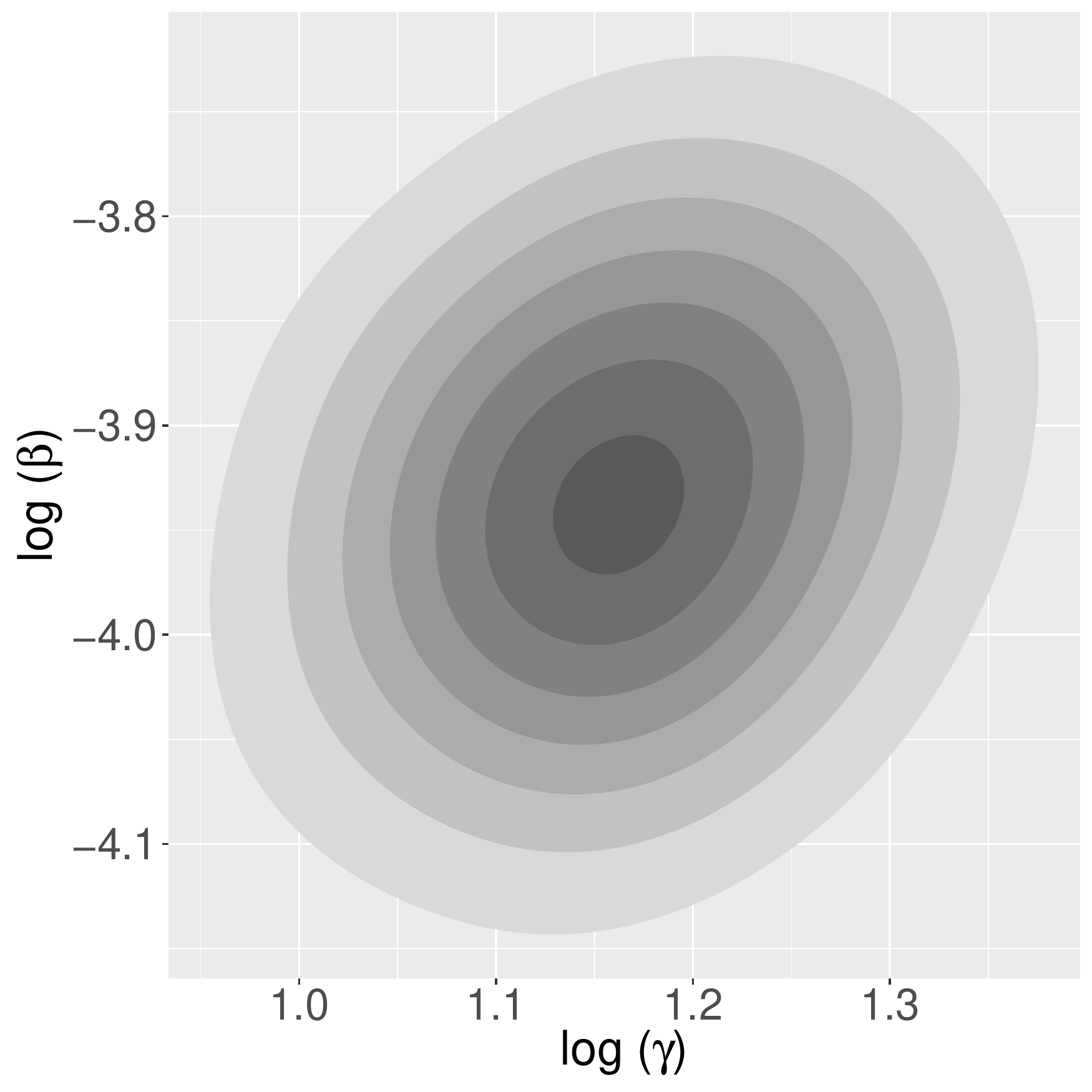}
\end{center}
\caption{Posterior density of infection $\beta$ and removal $\gamma$ rates
of the Eyam plague.}
\label{fig:HMC}
\end{figure}


\subsection{Adequacy of the classic SIR model}

Although the classic SIR model has been used extensively in practice,  it makes a strong assumption that each infected person can independently transmit the disease to one susceptible person with rate $\beta$. \citet{o2012modelling} argue that this assumption may not be realistic in settings where a saturation effect occurs; that is, a newly infected person contributes less to the overall infection pressure. Therefore, the authors propose to consider a general SIR model with infection rate $\beta S I^{\omega}$. This model is a special case of a more general SIR model where the infection rate is $\beta S^{\alpha} I^{\omega}$ and the removal rate is $\gamma I^{\eta}$ \citep{severo1969generalizations}.

\begin{figure}[h]
\centering
\begin{tikzpicture}
\shadedraw[ball color=black!50, opacity=0.4] (-0.25, 0) circle (0.8) node[opacity=1.0] {\Huge S};
\shadedraw[ball color=black!50, opacity=0.4] (3, 0) circle (0.8) node[opacity=1.0] {\Huge I};
\shadedraw[ball color=black!50, opacity=0.4] (6.25, 0) circle (0.8) node[opacity=1.0] {\Huge R};
%
%
\draw[->, line width=2pt] (0.75, 0) -- (1.5, 0) node[above] {{\Large $\beta S^{\alpha} I^{\omega}$ }} -- (2, 0);
\draw[->, line width=2pt] (4, 0) -- (4.75, 0) node[above] {{\Large $\gamma I^{\eta}$}} -- (5.25, 0);
\end{tikzpicture}
\caption{A directed graph representation of the general SIR model \citep{severo1969generalizations}.}
\label{fig:generalSIR}
\end{figure}
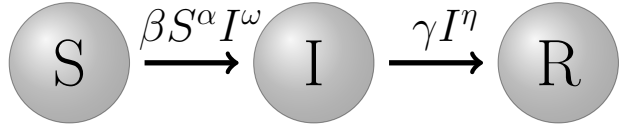

Our computational method does not require any special structure for the infection and removal rates, thus can also be applied to evaluate the likelihood function under these general SIR models.
In particular, $\{ N_{SI}(t), N_{IR}(t) \}$ in the general SIR model from \citet{severo1969generalizations} is a bivariate birth process with birth rates $\beta [(s_0 - N_{SI})^+]^\alpha [(i_0 + N_{SI} - N_{IR})^+]^\omega$ and removal rates $\gamma [(i_0 + N_{SI} - N_{IR})^+]^\eta$.
So, we can address some questions about model adequacy of the classic SIR model.
To illustrate, we use Bayes factors to assess if the classic SIR model is appropriate for the Eyam plague dynamics \citep{raggett1982}.
In particular, we test between the general SIR model against its nested sub-models.
Since the tests are between nested models, we apply the Savage-Dickey density ratio to evaluate the Bayes factors \citep{verdinelli1995computing}.
To be specific, if model $\mc{M}_0$ with parameter $(\theta = 0,\phi)$ is nested within model $\mc{M}_1$ with parameter $(\theta, \phi)$ and the prior $p_0(\phi)$ under $\mc{M}_0$ is proportional to the prior $p_1(\theta = 0, \phi)$ under $\mc{M}_1$, then the Bayes factor $B_{01}$ in favor of $\mc{M}_0$ over $\mc{M}_1$ can be estimated via the marginal posterior distribution under $\mc{M}_1$ as follows
\begin{equation}
B_{01} = \frac{p(\theta = 0 | \mb{Y}, \mc{M}_1)}{p(\theta=0 | \mc{M}_1)}
\end{equation}
where $p(\theta = 0 | \mb{Y}, \mc{M}_1)$ and $p(\theta = 0 | \mc{M}_1)$ are marginal posterior and prior densities of $\theta$ evaluated at $0$ under model $\mc{M}_1$.
Here, we posit independent log-normal priors $\ln \mc{N}(0, 100^2)$ for each parameter. Therefore, the condition for applying the Savage-Dickey density ratio is satisfied.
To estimate the posterior distribution under the general SIR model, we use our
MCMC tools.
The marginal posterior densities are estimated using kernel density estimation implemented in the R package \texttt{ks} \citep{duong2007ks} and these estimates are then used to compute the Bayes factors via the Savage-Dickey density ratio.
Table \ref{tab:bayesfactors} lists these Bayes factors, and we can see that they strongly support the classic SIR model over the general SIR model.
Although Savage-Dickey density ratio is not the best approximation method for Bayes factors, we can safely ignore its drawback because the evidence supporting the classic SIR model is overwhelming.

\begin{table}[h]
\centering
\begin{tabular} { lc }
\hline
  \multicolumn{1}{c}{Model $\mc{M}_0$} &
  \multicolumn{1}{c}{$\log_{10} B_{01}$} \\
\hline
  $\alpha = \omega = \eta = 1$ & 6.9 \\ 
  $\alpha = \omega = 1$        & 4.4 \\ 
  $\omega = \eta = 1$          & 4.7 \\ 
  $\alpha = \eta = 1$          & 4.6 \\ 
  $\alpha = 1$                 & 2.2 \\ 
  $\omega = 1$                 & 2.2 \\ 
  $\eta = 1$                   & 2.6 \\ 
  \hline
\end{tabular}
\label{tab:bayesfactors}
\caption{Bayes factors $B_{01}$ in favor of nested models $\mc{M}_0$ over the general SIR model $\mc{M}_1$ estimated using the Savage-Dickey density ratio.}
\end{table}



\section{Ebola outbreak in Guinea}
\label{sec:Ebola}

Ebola is a contagious viral hemorrhagic fever caused by \textit{Zaire ebolavirus}.
The fatality rate of Ebola is very high, up to $70.8 \%$ \citep{team2014ebola}. The 2014-2015 Ebola outbreak in West Africa is the largest Ebola epidemic in history. In this section, we focus on the outbreak in Guinea from January 2014 to May 2015 (73 weeks). During this period, the World Health Organization (WHO) has convened 5 meetings of the IHR Emergency Committee regarding the Ebola outbreak in West Africa. The first three meetings happened in three consecutive month August, September and October 2014. During the fourth meeting in January 2015, \citet{WHOstatement4th} noted that the number of Ebola cases in Guinea had decreased since the third meeting. \citet{team2015west} also confirmed that the Ebola outbreak has slowed down since October 2014.

We study this change in the trajectory of the outbreak using the number of reported Ebola cases reported weekly in $19$ prefectures across Guinea.
To be specific, we are interested in finding evidence that the outbreak in Guinea became less severe after the third WHO meeting and in what regions this happened. These $19$ prefectures are the only places in Guinea where Ebola cases were reported both before and after the third WHO meeting.

\newcommand{\pidx}{p}
\newcommand{\hierarchyMean}{\mathbf{M}}
\newcommand{\hierarchyVariance}{\boldsymbol{\Sigma}}
\newcommand{\hierarchyPrecision}{\mathbf{P}}

We employ a hierarchical and time-inhomogeneous, but still Markovian, SIR model to analyze these data. We re-parameterize the SIR model by replacing the infection rate $\beta$ with the basic reproduction number $R_0 := \beta N /\gamma$. Basic reproduction number is an important concept in epidemiology and can be interpreted as the average number of secondary infections caused by a new infectious individual in a susceptible population. When $R_0 < 1$ the disease will die out, and when $R_0 > 1$ the disease will be able to spread in the population. Researchers often use the value of $R_0$ to measure the severity of an epidemic. To simplify the analysis, we assume that the population of each prefecture is closed.
In other words, we na\"ively assume that the movement between prefectures and the movement in and out of Guinea are negligible.
This assumption is violated if large number of healthy persons or small number of infected persons enter (or leave) a prefecture.
We obtain the total populations of these prefectures from the 2014 census \url{https://en.wikipedia.org/wiki/Prefectures_of_Guinea}.

We use a ``week'' as the unit for time in this analysis.
Letting $t_0$ be the week when the third WHO meeting happened, our model proceeds as follows: the Ebola cases of each prefecture follow a conditionally independent SIR process with parameters $R_{0\pidx}(t)$ and $\gamma_{\pidx}$ for prefecture $\pidx = 1,\ldots,19$. Further, $R_{0\pidx}(t)$ is a time-inhomogeneous function that satisfies
\begin{equation}
\log R_{0\pidx}(t) = \log r_{0\pidx} + 1_{\{t \geq t_0\}} \log \delta_{\pidx},
\end{equation}
where $r_{0\pidx}$ quantifies the basic reproduction number before $t_0$ and $\delta_{\pidx}$ is the scale factor by which the basic reproduction number changes after $t_0$ in prefecture $\pidx$. Moreover, we assume a simple hierarchical prior distribution
\begin{equation}
\left(\log r_{0\pidx}, \log \delta_{\pidx},\log \gamma_{\pidx}\right)^{t}
\sim
\mc{N}
\left( \hierarchyMean, \text{diag}(\hierarchyVariance) \right),
\end{equation}
where $\hierarchyMean = (\mu_r, \mu_\delta, \mu_\gamma)$ is the grand-mean on the log-scale across prefectures and $\hierarchyVariance = (\sigma^2_r, \sigma^2_\delta, \sigma^2_\gamma)$ is the variance, with relatively uninformative conjugate hyperpriors
\begin{equation}
\begin{aligned}
\mu_{\phi} & \sim \mc{N} \left( \mathbf{0}, 10^2 \right),~\text{and}~
\sigma^2_{\phi} & \sim \text{InverseGamma} \left( 10^{-3}, 10^{-3} \right),~~ \phi \in \{ r,\delta,\gamma \}.
\end{aligned}
\end{equation}

%
%
%

Of primary scientific interest, $\delta_{\pidx} < 1$ corresponds to a reduction in the basic reproduction number in prefecture $\pidx$, suggesting that the Ebola outbreak slowed down in that prefecture.
However, an important limitation of the data arises, in that field epidemiologists were only able to record the number of new cases between time points.
The number of removals is unknown.
To overcome this limitation, we use a Metropolis-within-Gibbs scheme to sample the posterior distribution of the rate parameters and the number of removals (see Appendix \ref{MHwithinG} for more details).
Because we can compute the joint transition probability matrix between time points, we can draw directly from the full conditional distribution of the removal number, leading to substantially more efficient numerical integration than previous data augmentation approaches that require all sufficient statistics of the completely observed likelihood.
Further, we can speed up this sampling scheme by updating the unknown parameters in each prefecture in parallel. The result is summarized in Figure \ref{fig:mapEbola}, where we plot estimates of the basic reproduction number for each prefecture on the map of Guinea. Yellow circles represent $r_{0\pidx}$ and blue circles represent $r_{0\pidx} \times  \delta_{\pidx}$ when the posterior probability that $\delta_{\pidx} < 1$ is greater than 97.5\% \
Note that there is no posterior evidence supporting $\delta_{\pidx} > 1$ for any $p$ because the posterior probability that $\delta_{\pidx} > 1$ is less than $0.5$ for all $p$.
The radius of each circle reports a posterior mean estimate.
We present the posterior means and 95\% Bayesian credible interval of $\hierarchyMean$
and of $\hierarchyVariance$
in Table \ref{tab:hier}.


\begin{table}[h]
\begin{center}
\begin{tabular}{ crr@{ (}r@{, }l@{) }c@{}l }
\hline
&
\multicolumn{1}{c}{Posterior} &
\multicolumn{4}{c}{95\% Bayesian} \\
Parameter &
\multicolumn{1}{c}{mean} &
\multicolumn{4}{c}{credible interval}  \\
\hline
$\mu_r$ & ~7.47 $\times 10^{-2}$ && -0.425 & 15.1  && $\times 10^{-2}$ \\
$\mu_\delta$ & -1.25 $\times 10^{-1}$ && -2.36 & -0.0844 && $\times 10^{-1}$ \\
$\mu_\gamma$ & -6.76 $\times 10^{-1}$ && -10.4 & -3.19  && $\times 10^{-1}$ \\
$\sigma^2_r$ & ~4.67 $\times 10^{-3}$ && 0.399 & 24.1  && $\times 10^{-3}$ \\
$\sigma^2_\delta$ & ~2.24 $\times 10^{-2}$ && 0.216 & 8.18  && $\times 10^{-2}$ \\
$\sigma^2_\gamma$ & ~5.98 $\times 10^{-1}$ && 2.84 & 12.0  && $\times 10^{-1}$ \\
\hline
\end{tabular}
\end{center}
\caption{Posterior mean and 95\% Bayesian credible interval of hierarchical parameters $\hierarchyMean$ and $\hierarchyVariance$.}
\label{tab:hier}
\end{table}



\begin{figure}[h]
\centerline{%
\includegraphics[width=0.7\textwidth]{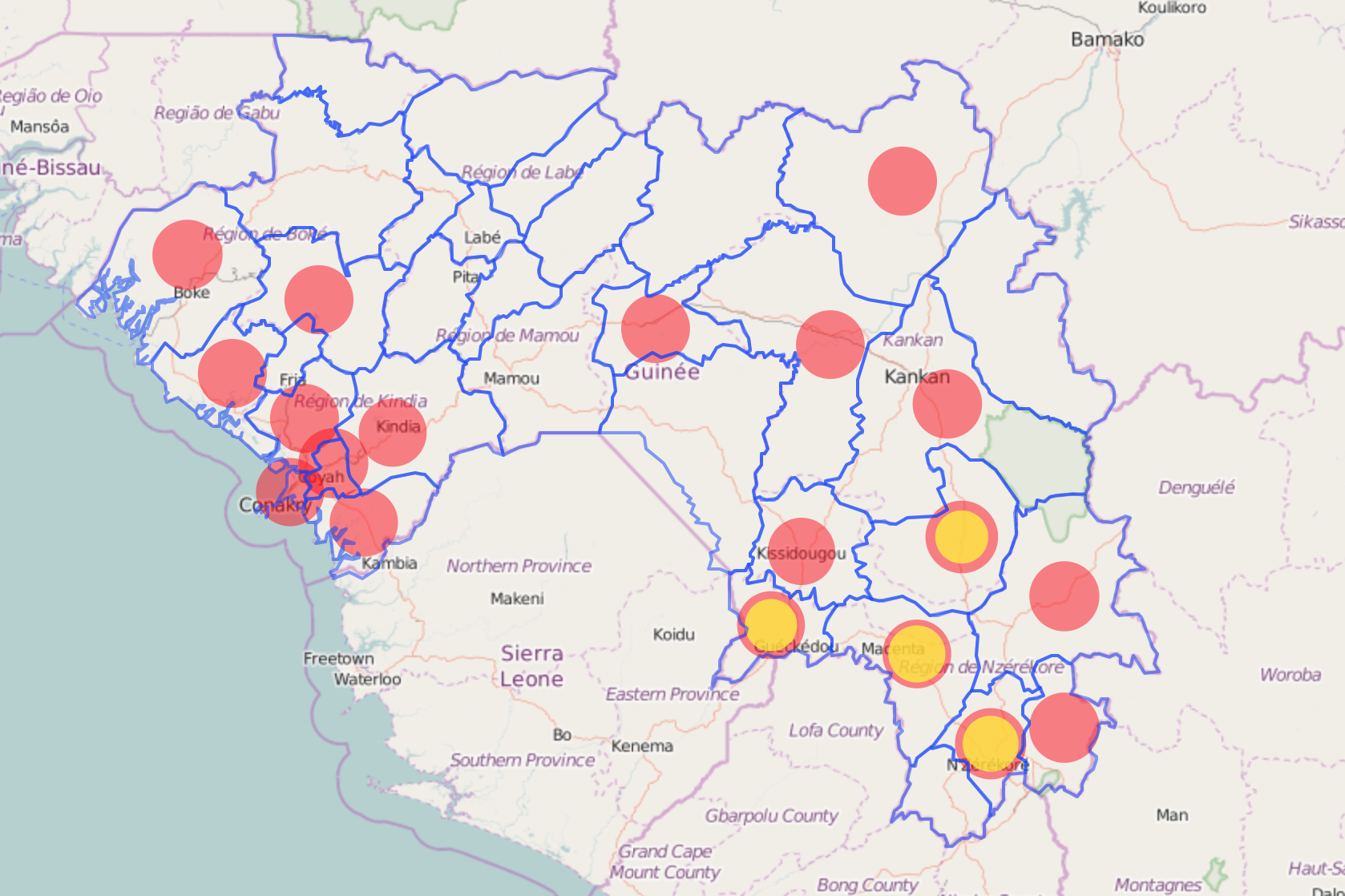}%
}
%
%
%

\vspace*{-1.25cm}\hspace*{2.5cm}
\scalebox{1.0}{%
\begin{tikzpicture}

\draw[fill, color=red, opacity=0.5] (0, 0) circle (0.175) node (node1) [right = 0.25, color=black] {{\color{blue} $r_{0\pidx}$}};
\draw[fill, color=yellow, opacity=0.7] (0, -0.5) circle (0.175) node (node2)  [right = 0.25, color=black] {{\color{blue} $r_{0\pidx} \times \delta_{\pidx}$ when $ \delta_{\pidx} \ne 1$ }};

\end{tikzpicture}
} 

\caption{Basic reproduction numbers for $19$ prefectures in Guinea before and after the third WHO meeting.}
\label{fig:mapEbola}
\end{figure}
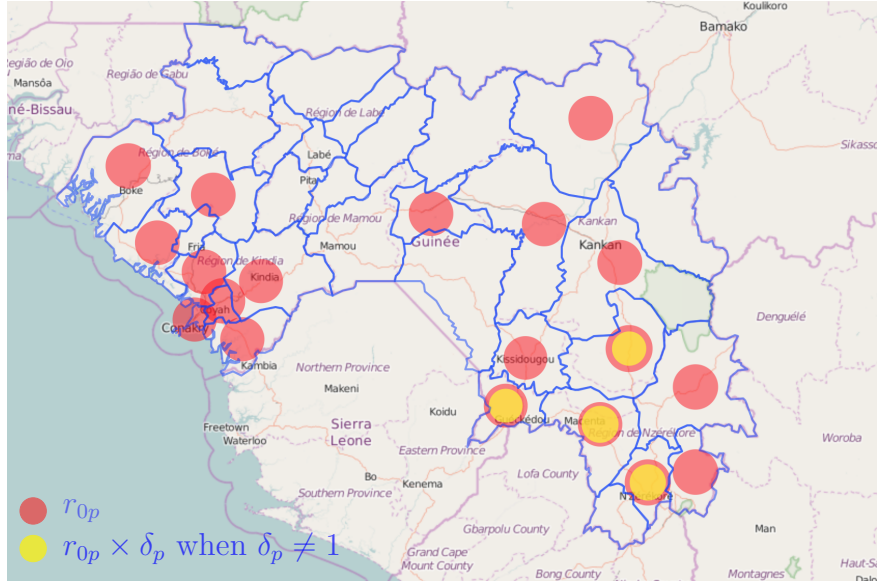

We note that \textit{a posteriori} $\sigma^2_\gamma$ is larger than $\sigma^2_r$ or $\sigma^2_\delta$, with probability approaching $1$.
Therefore, the removal rate $\gamma$ varies across the country more than the reproduction number $R_0$.
The posterior of $(\mu_\delta, \sigma^2_\delta)$ provides evidence for the slowing down of the Ebola outbreak in Guinea after the third WHO meeting.
However, Figure \ref{fig:mapEbola} suggests that the epidemic only slowed down in the Southeast region of Guinea while the epidemic in other regions seems to stay the same.
This finding gives a clearer picture of the change in the trajectory of the Ebola epidemic in Guinea.
It raises a very practical question: what made the outbreak in the Southeast region of Guinea to slow down?
Answering this question could help in efforts to find a more effective method for controlling Ebola epidemic.






\section{Discussion}
\label{sec:dis}

In this paper, we develop an algorithm to compute the transition probabilities of stochastic compartmental models for inference from surveillance data.
We introduce a new representation for compartmental processes using multivariate birth processes and, through this representation,
avoid the need for continued fraction evaluation to solve the Chapman-Kolmogorov equations.
With quadratic complexity in number of transitions between observations, our approach emerges as computationally more efficient than previous methods for the ubiquitous SIR model and applicable to a larger class of compartmental models, such as the SEIR and SIRS models.
Further performance gains through embarrassingly parallel evaluation of the series in Equation \eqref{eqn:invLap} remain open.

Since the formulation of the SIR model over 90 years ago, many have viewed its transition probabilities as beyond reach.
We provide some brief intuition on why the Laplace transform of the transition probabilities carries mere quadratic complexity
$\mc{O}(\Delta_S \Delta_R)$.
Viewed as a multivariate birth process that conveniently only increases, the transition probabilities we seek are related to the waiting time until the $\Delta_S$ and $\Delta_R$ births have occurred.
Inter-birth times are independent exponential random variables with potentially unique rates, and we can arrive at the distribution of the total waiting time through taking a convolution of $\Delta_S + \Delta_R$ of these exponentials.
However, the rates depend on the order of births and there are $(\Delta_S + \Delta_R)! / \Delta_S! \Delta_R!$ possible orderings.
Putting these pieces together, the transition probabilities are then exponential sums of multiple convolutions.
We recall several properties of Laplace transformations.  First, they are linear operators, so sums in probability-space remain sums in the transformed space.
Second, convolutions metamorphose into simple multiplication in the transformed space.
These properties leave us with a sum-product expression, suggesting a distribution of the sums within the products.
To gain insight into this dynamic programming, consider the $\Delta_S \times \Delta_R$ lattice graph.
Each lattice path from $(0,0)$ to $(\Delta_S, \Delta_R)$ represents one possible ordering of the birth events.
If we want the transformed probability of ending at $(\Delta_S, \Delta_R)$, there are only two possible one-shorter paths that could have gotten us there, specifically $(0,0)$ to $(\Delta_{S-1}, \Delta_R)$ or $(\Delta_S, \Delta_{R - 1})$.
So, the resulting transformed probability becomes the sum of the two shorter-path transformed probabilities, each multiplied by the Laplace transform of an exponential random variable that has a simple, closed-form expression.
Consequentially, in filling out the whole lattice graph, we need to visit each point once in increasing order and there are only $\Delta_S \Delta_R$ points.

Because differentiation is also a linear operator, our recursion method remains pertinent for computing the derivatives of the transition probabilities with respect to the unknown parameters of the compartmental model.
This feature makes HMC-based Bayesian inference feasible.
As the number of unknown parameters in the compartmental models grows, we suspect HMC to generally outperform Metropolis-Hastings algorithms using
standard
transition kernels.
Equally noteworthy, our algorithm does not require any specific structure in the birth rates $\lambda^{(\cdot)}_x$ of the multivariate birth processes. Therefore, we can apply our method to other general stochastic epidemic models such as one proposed by \citet{severo1969generalizations}.  This opens the possibility to access the model adequacy of traditional epidemic models.
It is worth noticing that our method only works for time-homogeneous rates between observations.
When the rates depend on time, the Chapman-Kolmogorov equations in the Laplace domain do not have analytic formulae making the current tool inapplicable. 
Therefore, an important subject for future direction of this work is extending to time-inhomogeneous processes.

Finally, we examine the 2014-2015 Ebola outbreak in Guinea using a marginalized, hierarchical and time-inhomogeneous Markovian SIR model.
By applying our recursion method, we can effectively explore the posterior distribution of the basic reproductive number and removal rate across the country, while simultaneously integrating out the unobserved removed population sizes using a Metropolis-within-Gibbs scheme.
This example highlights the flexibility of a Bayesian framework for direct likelihood-based inference for a compartmental model when
one or more of the compartments are missing or immeasurable, as is common in infectious disease surveillance.
Our results provide evidence for the slowing down of this epidemic in the Southeast region of Guinea.
Several important extensions are immediately obvious.
For example, we assume no error in the reported Ebola case counts, but a simple modification similar to that we accomplished for missing compartments can relax this assumption.

\section*{Acknowledgments}

This work was partially supported by the National Institutes of Health (R01 HG006139, R01 AI107034, KL2 TR000140, P30MH062294, and DP2 OD022614-01) and the National Science Foundation (IIS 1251151 and DMS 1264153).


\appendix

\section{Derivatives of the transition probabilities of SIR model}
\label{derivatives}

We propose an efficient method to evaluate the derivatives of the transition probabilities of the SIR model. Again, we use the bivariate birth presentation for this model. Denote $X = N_{SI}$ and $Y = N_{IR}$, and consider the forward transition probability $P_{xy}(t) = \Pr \{ X(t) = x, Y(t) = y ~|~ X(0) = 0, Y(0) = 0 \}$. The forward Chapman-Kolmogorov equations are:
\begin{equation}
\begin{aligned}
\frac{dP_{xy}(t)}{dt} = & \beta (s_0 - x + 1)^+ (i_0 + x - 1 - y)^+ P_{x-1, y}(t) \\
&+ \gamma(i_0 + x - y + 1)^+ P_{x, y-1}(t) \\
&- [\beta (s_0 - x)^+ (i_0 + x - y)^+ + \gamma(i_0 + x - y)^+] P_{xy}(t).
\end{aligned}
\label{eq:forward2d}
\end{equation}
Let $P^{(\beta)}_{xy}$ be the derivative of $P_{xy}$ with respect to $\beta$. From \eqref{eq:forward2d}, we have
\begin{equation}
\begin{aligned}
\frac{dP^{(\beta)}_{xy}(t)}{dt} = & \beta (s_0 - x + 1)^+ (i_0 + x - 1 - y)^+ P^{(\beta)}_{x-1, y}(t) \\
&+ \gamma(i_0 + x - y + 1)^+ P^{(\beta)}_{x, y-1}(t) \\
&- [\beta (s_0 - x)^+ (i_0 + x - y)^+ + \gamma(i_0 + x - y)^+] P^{(\beta)}_{xy}(t) \\
&+  (s_0 - x + 1)^+ (i_0 + x - 1 - y)^+ P_{x-1, y}(t) \\
&- (s_0 - x)^+ (i_0 + x - y)^+ P_{xy}(t)
\end{aligned}
\label{eq:forward2d_diff}
\end{equation}
Denote $f_{xy}$ and $f^{(\beta)}_{xy}$ be the Laplace transform of $P_{xy}$ and $P^{(\beta)}_{xy}$ respectively. Taking Laplace transform to both sides of \eqref{eq:forward2d_diff}, we have
\begin{equation}
\begin{aligned}
sf^{(\beta)}_{xy}(s) - P^{(\beta)}_{xy}(0) = & \beta (s_0 - x + 1)^+ (i_0 + x - 1 - y)^+ f^{(\beta)}_{x-1, y}(s) \\
&+ \gamma(i_0 + x - y + 1)^+ f^{(\beta)}_{x, y-1}(s) \\
&- [\beta (s_0 - x)^+ (i_0 + x - y)^+ + \gamma(i_0 + x - y)^+] f^{(\beta)}_{xy}(s) \\
&+  (s_0 - x + 1)^+ (i_0 + x - 1 - y)^+ f_{x-1, y}(s) \\
&- (s_0 - x)^+ (i_0 + x - y)^+ f_{xy}(s) .
\end{aligned}
\label{eq:Laplace2d}
\end{equation}
Since $P_{xy}(0) = 1_{\{x=0,y=0\}}$ for all $\beta$, we deduce that $P^{(\beta)}_{xy}(0) = 0$. Therefore, we can compute $f^{(\beta)}_{xy}$ using the following recursion
\begin{equation}
\begin{aligned}
f^{(\beta)}_{00}(s) &= - \frac{s_0 i_0 f_{00}(s)}{s + \beta s_0 i_0 + \gamma i_0} \\
f^{(\beta)}_{xy}(s) &= \frac{\beta (s_0 - x + 1)^+ (i_0 + x - 1 - y)^+ f^{(\beta)}_{x-1, y}(s)}{s + \beta (s_0 - x)^+ (i_0 + x - y)^+ + \gamma(i_0 + x - y)^+} \\
&+ \frac{\gamma(i_0 + x - y + 1)^+ f^{(\beta)}_{x, y-1}(s)}{s + \beta (s_0 - x)^+ (i_0 + x - y)^+ + \gamma(i_0 + x - y)^+} \\
&+ \frac{(s_0 - x + 1)^+ (i_0 + x - 1 - y)^+ f_{x-1, y}(s)}{s + \beta (s_0 - x)^+ (i_0 + x - y)^+ + \gamma(i_0 + x - y)^+} \\
&- \frac{(s_0 - x)^+ (i_0 + x - y)^+ f_{xy}(s)}{s + \beta (s_0 - x)^+ (i_0 + x - y)^+ + \gamma(i_0 + x - y)^+}.
\end{aligned}
\end{equation}
Then, we can compute $P^{(\beta)}_{xy}$ by approximating the inverse Laplace transform using \eqref{eqn:invLap}.
Similarly, we can derive the recursive formulae for $f^{(\gamma)}_{xy}$:
\begin{equation}
\begin{aligned}
f^{(\gamma)}_{00}(s) &= - \frac{i_0 f_{00}(s)}{s + \beta s_0 i_0 + \gamma i_0}  \\
f^{(\beta)}_{xy}(s) &= \frac{\beta (s_0 - x + 1)^+ (i_0 + x - 1 - y)^+ f^{(\beta)}_{x-1, y}(s)}{s + \beta (s_0 - x)^+ (i_0 + x - y)^+ + \gamma(i_0 + x - y)^+} \\
&+ \frac{\gamma(i_0 + x - y + 1)^+ f^{(\beta)}_{x, y-1}(s)}{s + \beta (s_0 - x)^+ (i_0 + x - y)^+ + \gamma(i_0 + x - y)^+} \\
&+ \frac{(i_0 + x - y + 1)^+ f_{x, y-1}(s)}{s + \beta (s_0 - x)^+ (i_0 + x - y)^+ + \gamma(i_0 + x - y)^+} \\
&- \frac{(i_0 + x - y)^+ f_{xy}(s)}{s + \beta (s_0 - x)^+ (i_0 + x - y)^+ + \gamma(i_0 + x - y)^+},
\end{aligned}
\end{equation}
and evaluate $P^{(\gamma)}_{xy}$ using \eqref{eqn:invLap}.


\section{Metropolis-within-Gibbs algorithm for inference of Ebola dynamics in West Africa}
\label{MHwithinG}

Let $\mb{t}^{(\pidx)} = (t^{(\pidx)}_1, t^{(\pidx)}_2, \ldots, t^{(\pidx)}_{{\idx}_{\pidx}})$ be the times when the counts of Ebola cases in prefecture $\pidx$ are reported. We define $\mb{N}_{SI}^{(\pidx)}$ and $\mb{N}_{IR}^{(\pidx)}$ be the total numbers of new infection and removal events at $\mb{t}_{-1}^{(\pidx)}$ respectively. Here, $\mb{t}_{-j}^{(\pidx)}$ denotes the vector $\mb{t}^{(\pidx)}$ without the $j^{\mbox{\tiny th}}$ coordinate. Notice that we only observe the total of Ebola cases at $\mb{t}^{(\pidx)}$, thus we only know $\mb{N}_{SI}^{(\pidx)}$. So, our unknown parameters are $\{ \mb{N}_{IR}^{(\pidx)}, r_{0\pidx}, \delta_{\pidx}, \gamma_{\pidx} \}$ for all $p$ and $( \hierarchyMean, \hierarchyVariance ) $. We update our parameters using a Metropolis-within-Gibbs algorithm as follows:

\begin{enumerate}
	\item For every $\pidx = 1,\ldots, 19$ in parallel,
	\begin{enumerate}[label=(\roman*)]
		\item For every $j = 2, 3, \ldots, {\idx}_{\pidx}$, we can compute $ \mb{P}(\mb{N}_{IR}^{(\pidx)}(t_j) = n ~|~ \mb{N}_{SI}^{(\pidx)}, \mb{N}_{IR}^{(\pidx)}(\mb{t}^{(\pidx)}_{-j}), r_{0\pidx}, \delta_{\pidx}, \gamma_{\pidx})$ using the forward and backward transition probabilities of the SIR model. Therefore, we sample from $\mb{N}_{IR}^{(\pidx)}(t_j) ~|~ \mb{N}_{SI}^{(\pidx)}, \mb{N}_{IR}^{(\pidx)}(\mb{t}^{(\pidx)}_{-j}), r_{0\pidx}, \delta_{\pidx}, \gamma_{\pidx}$ directly to update the value of $\mb{N}_{IR}^{(\pidx)}(t_j)$.
		\item Then, we update $r_{0\pidx}, \delta_{\pidx}, \gamma_{\pidx} ~|~ \mb{N}_{SI}^{(\pidx)},  \mb{N}_{IR}^{(\pidx)}$ on the $\log$-scale using a random-walk Metropolis-Hasting algorithm with Gaussian proposals or HMC. This step is straight forward because we can evaluate the density $l(r_{0\pidx}, \delta_{\pidx}, \gamma_{\pidx} ~|~ \mb{N}_{SI}^{(\pidx)},  \mb{N}_{IR}^{(\pidx)})$ efficiently.
	\end{enumerate}

	\item Finally, since we choose conjugate priors for the hierarchical parameters, we Gibbs sample $\hierarchyMean$ and $\hierarchyVariance$ .
\end{enumerate}

Note that we update $\mb{N}_{IR}^{(\pidx)}(t_j)$ sequentially instead of sampling from the joint distribution of $\mb{N}_{IR}^{(\pidx)}$ because sampling sequentially only requires transition probability matrices between counts of Ebola cases, which is much smaller compared to the full transition probability matrix of size $N^2 \times N^2$, where $N$ is the total population, required for sampling from the joint distribution.


\bibliographystyle{chicago}
\bibliography{ms}
\end{document}